\documentclass[10pt]{amsart}
\usepackage{amssymb,latexsym}
\theoremstyle{plain}
\numberwithin{equation}{section}
\newtheorem{theorem}{Theorem}
\newtheorem{lemma}{Lemma}
\newtheorem{proposition}{Proposition}

\newtheorem{remark}{Remark}
\newtheorem{definition}{Definition}
\begin{document}

\title[ Quantum Feedback  Control]
{\textbf {Application of Quantum Stochastic Calculus to Feedback   Control}}

\author{Andreas Boukas}
\address{Department of Mathematics and Natural Sciences, American College\\
 Aghia Paraskevi, Athens 15342, Greece}
\email{andreasboukas@acgmail.gr}

\date{}

\begin{abstract} The basic  aspects of the Hudson-Parthasarathy quantum stochastic calculus and of the Accardi-Fagnola-Quaegebeur representation free stochastic calculus are presented. The basic features of the stochastic calculus for the square of white noise recently developed by Accardi-Boukas are  described. The  feedback control problem for stochastic processes driven by quantum noise is  solved
\end{abstract}

\maketitle

gjpam.tex

\section{Introduction} 

 The time-evolution of a dynamical system is described by some kind of differential equation. To be useful, for example for engineering purposes, the solution of such a differential equation must be under the control of the engineer. Typically, the size of the solution of a deterministic, classical stochastic, or quantum stochastic differential equation, is controlled by choosing an input so as to minimize an appropriate performance (or cost)  functional. In classical engineering, control theory and in particular its quadratic aspect has found many applications, for example in space navigation and flight technology. The problem of controlling quantum stochastic evolutions arises naturally in several different fields such as quantum chemistry, quantum information theory and  quantum engineering. 

 The case of dynamical systems described by linear differential equations, with performance controlled by minimizing a quadratic  functional,  is of particular importance  since in this case the control process can be explicitly calculated and is given in terms of a feedback law.  Feedback controls are important because, by taking into account the state of the system at each moment, they allow for updated and stabilizing control action. 

In these notes we review some basic aspects of the classical theory of quadratic control, we briefly describe quantum stochastic calculus and we present our results in the quantum case, obtained over the past few years.  In particular we apply the recently discovered closed form of the unitarity conditions for stochastic evolutions driven by the square of white noise (see \cite{ 2}) to describe the solution of the quadratic cost control problem in that case.

In the classical case, the feedback law that describes the optimal control is defined in terms of the solution of a scalar-, matrix-, or operator-valued Riccati differential equation, which in the case of time-independent coefficients reduces to an algebraic Riccati equation.  In the quantum case the equation is replaced by a quantum stochastic Riccati equation which  in the case of time-independent coefficients reduces to an operator algebraic Riccati equation as in \cite{9}.  The quadratic form of the control criterion allows the quantum control problem to be solved in analogy with the classical stochastic control problem with the use of quantum stochastic calculus. Under general conditions on the system Hamiltonian part of the stochastic evolution and on the system observable to be controlled, the operator algebraic Riccati equation admits solutions with the required properties which can also be explicitly described. 

Related to the control problem of quantum stochastic evolutions, in the quantum case,  is the problem of optimal control of the solution of a quantum Langevin equation with constant coefficients which naturally arises in several different fields such as quantum chemistry, quantum information, quantum engineering e.t.c.  The mathematical formulation of this problem was recently considered in \cite{3} for quantum systems affected by first order and square of white noise. It was preceded by several studies on the quadratic control of the solution of a quantum evolution  driven by first order white noise (see references in \cite{3}) and the dual Kalman--Bucy filtering problem (see \cite{1}). The statement of the problem is the following: one starts from the Langevin equation for a system observable $X$ and then looks for the coefficients of the martingale terms of this equation, which minimize a given quadratic cost functional. Once this problem is solved, using the stochastic limit technique of \cite{6}, one then looks for a bona fide Hamiltonian interaction which, in the stochastic limit, gives rise to the optimal Langevin equation. Thus, combining techniques of quantum stochastic control with the stochastic limit of quantum theory, one can find a real physical interaction to concretely realize the optimal quantum evolution. 

These lecture notes are organized as follows: In section 2 we discuss the basic concepts and results of classical control theory, both deterministic and stochastic, in particular those related to quadratic performance criteria and linear systems, as in \cite{10} and \cite{11}. In section 3 we describe quantum stochastic calculus, associated with the first order Hida white noise functionals as in \cite{7} and \cite{13}. In section 4 we describe square of white noise quantum stochastic calculus. In section 5 we describe the representation free quantum stochastic calculus of Accardi-Fagnola-Quaegebeur  of \cite{5}. In section 6 we discuss the emergence of quantum stochastic differential equations from white noise equations through the stochastic limit of \cite{6}. In section 7 we describe the solution of the quantum quadratic control problem in the framework of sections 3, 4, and 5. Finally in section 8 we study quantum stochastic Riccati equations within the framework of the  representation free quantum stochastic calculus of section 5,

\section{Classical Linear Control} 

\subsection{Deterministic Control}

In the classical deterministic case we consider a system whose evolution over a finite time interval is modelled by the solution $x=\{x_t:\,t\in[0,T]\}\in C([0,T],\mathbb{R}^n)\}$ of an ordinary differential equation of the form 
 
\begin{eqnarray}
&dx_t=(A\,x_t+u_t)\,dt&\\
&x_0=x,\,\,t\in[0,T]&
\end{eqnarray}

where $A\in\mathcal{B}(\mathbb{R}^n)$, the space of bounded linear operators on $\mathbb{R}^n$, and  $u\in L_{\infty}([0,T],\mathbb{R}^n)$ or $L_2([0,T],\mathbb{R}^n)$. Although we consider here the finite-dimensional case, the concepts and the results can be extended from $\mathbb{R}^n$ to any Hilbert space $\mathcal{H}$.

We assume that we can interfere with the performance of the model by choosing the "control process" $u=\{u_t:\,t\in[0,T]\}$ so as to minimize a certain "performance (or cost) functional" $J(u)$. There is a wide variety of such functionals  designed for specific  models. However, the most computationally accessible one is the "quadratic" performance functional of the form  

\begin{eqnarray}
&J(u)=\int_0^T\,(<x_t,Q\,x_t>+<u_t,u_t>)\,dt+<x_T,\Pi\,x_T>&
\end{eqnarray}

where $<\cdot,\cdot>$ denotes the usual  inner product in the Euclidean space $\mathbb{R}^n$, $\Pi\in\mathcal{B}(\mathbb{R}^n)$,  $Q\in\mathcal{B}(\mathbb{R}^n)$, $\Pi\geq0$, $Q\geq0$. If the size of $x=\{x_t:\,t\in[0,T]\}$ is small, the performance functional (2.3) can serve as an approximation to many other functionals which are more adapted to the specific problem considered but also more computationaly complex. 
Before one looks for the optimal  control process, the system to be controlled must be "observable", "controllable", and "stabilizable". The definition of these concepts is as follows:

 \textbf{Observability}: Since the state $x_t$ of the system may only be accessible through an observation process $y_t=P\,x_t$, where $P\in\mathcal{B}(\mathbb{R}^n)$, we must be able to recreate $x_t$ (or equivalently $x_0$) from $y_t$.  If $P=1$ we speak of a "completely observable" system. Otherwise the system is only "partially observable".

\textbf{Controllability}: Given an initial state  $x_0$, we should be able to choose the control process $u=\{u_t:\,t\in[0,T]\}$ so that the system  will be steered in a finite time $t_1\in[0,T]$ to a desired state $x_1$. 

\textbf{Stabilizability}: In order to consider large terminal times $T$, we need the system to exhibit good long-run behavior i.e to eventually settle down to some steady-state behavior. From the mathematical point of view, this amounts to the asymptotic stability of the initial state of $x=\{x_t:\,t\in[0,T]\}$ or, equivalently, to the existence of a "feedback" control $u_t=K\,x_t$, where $K\in\mathcal{B}(\mathbb{R}^n)$, for which the system (2.1)-(2.2) is asymptotically stable. 

The performance functional (2.3) is particularly useful in the case when a system must operate at or near a particular state, chosen here to be the origin. We can think of the first term of (2.3) as a penalty for being too far away from the origin on $(0,T)$, the second as a penalty for using too much control and the third as a penalty for being too far away from the target at the final time $T$. The main result in the completely observable, classical  case is the following:

\begin{theorem} The performance functional (2.3) associated with the system (2.1)-(2.2) is minimized by the feedback control process

\begin{eqnarray}
u_t=-\Pi_t\,x_t
\end{eqnarray}

where $\{\Pi_t:t\in [0,T]\}$ is the solution of the Riccati differential equation

\begin{eqnarray}
&\frac{d}{dt}\Pi_t+A^*\,\Pi_t+\Pi_t\,A+Q-\Pi_t^2=0&\\
&\Pi_T=\Pi.&
\end{eqnarray}

\end{theorem}

 If we restrict to $u_t=-K\,x_t$, i.e to feedback controls with a time-independent coefficient,  then equations (2.5)-(2.6) are replaced by the "algebraic" Riccati equation

\begin{eqnarray}
A^*\,\Pi+\Pi\,A+Q-\Pi^2=0.
\end{eqnarray}

 \subsection{Stochastic Control}

In this case we consider systems whose time evolution is affected by noise. We assume that the noise can be accurately described by Brownian motion. Specifically, we consider systems whose time-evolution is described by the solution $x=\{x_t:\,t\in[0,T]\}$ of a stochastic differential equation of the form

\begin{eqnarray}
&dx_t=(A\,x_t+u_t)\,dt+C\,dB_t&\\
&x_0=x,\,\,t\in[0,T]&
\end{eqnarray}

where $A$ and $u$ are as in (2.1)-(2.2) with the added assumtion that $u$ is a stochastic process, $C\in\mathcal{B}(\mathbb{R}^n)$, and $B=\{B_t:t\geq0\}$ is a vector (in this case $n$-dimensional) Brownian motion. The performance functional (2.3) takes the form

\begin{eqnarray}
&J(u)=E(\int_0^T\,(<x_t,Q\,x_t>+<u_t,u_t>)\,dt+<x_T,\Pi\,x_T>)&
\end{eqnarray}

where $E$ denotes mathematical expectation. 

For completely observable systems, Theorem 1 remains true in the stochastic case. For partially observable systems, i.e when $x=\{x_t:\,t\in[0,T]\}$ is available only through an observation process $y=\{y_t:\,t\in[0,T]\}$ satisfying 

\begin{eqnarray}
&dy_t=H\,x_t\,dt+dW_t&
\end{eqnarray}

where $H\in\mathcal{B}(\mathbb{R}^n)$, and $W=\{W_t:t\in[0,T]\}$  is a vector (in this case $n$-dimensional) Brownian motion independent of $B=\{B_t:t\in[0,T]\}$, the main result  is as follows:

\begin{theorem} The performance functional (2.10) associated with the system (2.8)-(2.9) and (2.11) is minimized by the feedback control process

\begin{eqnarray}
u_t=-\Pi_t\,\hat{x_t}
\end{eqnarray}

where $\{\Pi_t:t\in [0,T]\}$ is the solution of the Riccati equation

\begin{eqnarray}
&\frac{d}{dt}\Pi_t+A^*\,\Pi_t+\Pi_t\,A+Q-\Pi_t^2=0&\\
&\Pi_T=\Pi&
\end{eqnarray}

and $\hat{x}=\{\hat{x_t}:t\in[0,T]\}$ is the minimum mean-square estimate of $x=\{x_t:\,t\in[0,T]\}$ given $\{y_s:s \leq t\}$, obtained through the Bucy-Kalman filter.
\end{theorem}

Since a dynamical system can be controlled by more than one competing controllers, where one wants to maximize a performance functional and another wants to minimize it, control problems can also be considered in the context of differential games as in \cite{12}.

\section{Quantum stochastic calculus} 

Let $B_t=\{B_t(\omega)/\,\omega\in \Omega\}$, $t\geq 0$,  be one-dimensional  Brownian motion.  Integration with respect to $B_t$ was defined by It\^{o}.  A basic result of the theory is that stochastic integral equations of the form

\begin{eqnarray}
&X_t=X_0+\int_0^tb(s,X_s)\,ds+\int_0^t\sigma (s,X_s)\,dB_s&
\end{eqnarray}

can be interpreted as stochastic differential equations of the form

\begin{eqnarray}
&dX_t=b(t,X_t)\,dt+\sigma (t,X_t)\,dB_t&
\end{eqnarray}

where  differentials are handled with the use of It\^{o}'s formula

\begin{eqnarray}
&(dB_t)^2=dt,\,\,\,\,\,\,dB_t\,dt=dt\,dB_t=(dt)^2=0&
\end{eqnarray}

In \cite{13a} Hudson and Parthasarathy obtained a Fock space representation of Brownian motion and Poisson process as follows.

\begin{definition}  The Boson Fock space  $\Gamma=\Gamma(L^2(\mathbb{R_+},\mathcal{C}))$  over    $L^2(\mathbb{R_+},\mathcal{C}) $  is the Hilbert space completion of the linear span of the exponential vectors  $\psi(f)$ under the inner product
\end{definition}

\begin{eqnarray}
<\psi(f),\psi(g)>=e^{<f,g>}
\end{eqnarray}

where $f,g \in L^2(\mathbb{R_+},\mathcal{C})$ and $<f,g>=\int_0^{+\infty}\,\bar{f}(s)\,g(s)\,ds$ where, here and in what follows,  $\bar{z}$ denotes the complex conjugate of $z\in \mathbb{C}$.
.

The annihilation, creation and conservation operators $A(f)$, $A^{\dagger}(f)$ and $\Lambda(F)$ respectively, are defined on the exponential vectors $\psi(g)$ of $\Gamma$ as follows.

\begin{definition}   

\begin{eqnarray}
&A_t\psi(g)=\int_0^t\,g(s)\,ds\,\,\psi(g)&\\
&A^{\dagger}_t\psi(g)=\frac{\partial}{\partial \epsilon}|_{\epsilon=0}\,\psi(g+\epsilon \chi_{[0,t]})&\\
&\Lambda_t\psi(g)=\frac{\partial}{\partial \epsilon}|_{\epsilon=0}\,\psi(e^{\epsilon \chi_{[0,t]})}g)&
\end{eqnarray}

\end{definition}

The basic   quantum stochastic differentials $dA_t$, $dA^{\dagger}_t$, and $d\Lambda_t$ are defined as follows.

\begin{definition} 

\begin{eqnarray}
&dA_t=A_{t+dt}-A_t& \\
&dA^{\dagger}_t = A^{\dagger}_{t+dt}-A^{\dagger}_{t}  &\\
&d \Lambda_t= \Lambda_{t+dt} -\Lambda_t &
\end{eqnarray}

\end{definition} 

Hudson and Parthasarathy defined stochastic integration with respect to the noise differentials of Definition 3 and obtained the  It\^{o} multiplication table

\begin{center}
\begin{tabular}{c|cccc}
$\cdot$&$dA_t^{\dagger}$&$d\Lambda_t$&$dA_t$&$dt$\\
\hline
$dA_t^{\dagger}$&$0$&$0$&$0$&$0$\\
$d\Lambda_t$&$dA_t^{\dagger}$&$d\Lambda_t$&$0$&$0$\\
$dA_t$&$dt$&$dA_t$&$0$&$0$\\
$dt$&$0$&$0$&$0$&$0$
\end{tabular}
\end{center}

The fundamental theorems of the Hudson-Parthasarathy quantum stochastic calculus  give formulas for expressing the matrix elements of quantum stochastic integrals  in terms of ordinary Riemann-Lebesgue integrals. 

\begin{theorem}
Let 

\begin{eqnarray}
&M(t)=\int_0^t\,E(s)\,d\Lambda (s)+F(s)\,dA(s)+G(s)\,dA^{\dagger}(s)+H(s)\,ds&
\end{eqnarray}

where $E$, $F$, $G$, $H$ are (in general) time dependent adapted processes. Let also $u\otimes \psi (f)$ and $v \otimes \psi (g)$ be in the exponential domain of $\mathcal{H} \otimes \Gamma$. Then

\begin{eqnarray}
&<u\otimes \psi (f),M(t)\,  v \otimes \psi (g)>=&\\
&\int_0^t <u \otimes \psi (f),\left(\bar{f}(s)\,g(s) \,E(s)+g(s)\,F(s)+\bar{f}(s)\,G(s)+H(s)\right) v \otimes \psi (g)>\,ds&\nonumber
\end{eqnarray}

\end{theorem} 

\begin{proof}

See theorem 4.1 of \cite{13a}

\end{proof}

\begin{theorem}
Let 

\begin{eqnarray}
&M(t)=\int_0^t\,E(s)\,d\Lambda (s)+F(s)\,dA(s)+G(s)\,dA^{\dagger}(s)+H(s)\,ds&
\end{eqnarray}

and

\begin{eqnarray}
&M^{\prime}(t)=\int_0^t\,E^{\prime}(s)\,d\Lambda (s)+F^{\prime}(s)\,dA(s)+G^{\prime}(s)\,dA^{\dagger}(s)+H^{\prime}(s)\,ds&
\end{eqnarray}

where $E$, $F$, $G$, $H$, $E^{\prime}$, $F^{\prime}$, $G^{\prime}$, $H^{\prime}$ are (in general) time dependent adapted processes. Let also $u\otimes \psi (f)$ and $v \otimes \psi (g)$ be in the exponential domain of $\mathcal{H} \otimes \Gamma$. Then

\begin{eqnarray}
&<M(t) \,u \otimes \psi (f),M^{\prime} (t)\, v \otimes \psi (g)>=&\\
&\int_0^t \{ <M(s)\,u \otimes \psi (f),\left( \bar{f}(s)\,g(s) E^{\prime} (s)+g(s)\,F^{\prime} (s)+\bar{f}(s)\,G^{\prime} (s)+H^{\prime} (s) \right) v \otimes \psi (g)>&\nonumber\\
&  + < \left( \bar{g}(s)\,f(s)\, E(s) + f(s)\,F(s) +\bar{g}(s)\, G(s)+H(s) \right) u \otimes \psi (f),  M^{\prime} (s)\,v \otimes \psi (g)>&\nonumber\\
&+   < \left( f(s) E(s)+G(s)   \right) u \otimes \psi (f),  \left( g(s) E^{\prime} (s)+G^{\prime} (s)   \right)  v \otimes \psi (g)>\}\, ds &\nonumber
\end{eqnarray}

\end{theorem} 

\begin{proof}

See theorem 4.3 of \cite{13a}

\end{proof}

The following lemma will be useful in deriving the fundamental result which connects classical with quantum stochastics.

\begin{lemma}[The first order Poisson-Weyl operator] 

Let 

\begin{eqnarray}
U(t)=e^{iE(t)}
\end{eqnarray}

 where 

\begin{eqnarray}
E(t)=\lambda \, t+z\,A_t+\overline{z}\,A^{\dagger}_t+k\,{\Lambda}_t
\end{eqnarray}

 with    $\lambda , k \in \mathbb{R}$,  $z \in \mathbb{C}$. 

(a) If $k \neq 0$ then

\begin{eqnarray}
&&\\
&dU(t)=U(t)\left[\left(i\lambda +\frac{|z|^2}{k^2} \,M\right)\,dt+\left(iz+\frac{z}{k}\,M\right)\,dA_t+\left(i\overline{z}+
\frac{\overline{z}}{ k}\,M\right)\,dA^{\dagger}_t+\left(i\,k+M\right)\,d{\Lambda}_t\right]&\nonumber
\end{eqnarray}

where

\begin{eqnarray}
M=e^{ik}-1-ik
\end{eqnarray}

(b) If $k=0$ then

\begin{eqnarray}
dU(t)=U(t)\left[\left(i\,\lambda-\frac{|z|^2}{2}\right)\,dt+i\,z\,dA_t+i\,\overline{z}\,dA^{\dagger}_t\right]
\end{eqnarray}

\end{lemma}

\begin{proof}

Computing the differential of $U(t)$ we find

\begin{eqnarray}
dU(t)&=&d(e^{i \, E(t)})\\
     &=& e^{i \, E(t+dt)} - e^{i \, E(t)}\nonumber\\
     &=&e^{i \, \left(E(dt)+E(t)\right)} - e^{i \, E(t)}\nonumber\\
     &=&e^{i \, E(dt)}\,e^{i \, E(t)} - e^{i \, E(t)}\mbox{ ( since of $\left[E(dt),E(t)\right]=0$) }\nonumber\\
     &=&e^{i \, E(t)}\,\left[e^{i \, dE(t)}-I\right]\nonumber\\
     &=&U(t)\,\sum_{n=1}^{\infty} \frac{\left(i\,dE(t)\right)^n}{n!}\nonumber
\end{eqnarray}

 Using the It\^{o} multiplication table we find that for $k \neq 0$ and $n \geq 2$

\begin{eqnarray}
dE(t)^n=|z|^2\, k^{n-2}\,dt+z\,k^{n-1}\,dA_t+\overline{z}\,k^{n-1}
dA^{\dagger}_t+k^n\,d{\Lambda}_t
\end{eqnarray}

while for $k=0$ and $n \geq 2$

\begin{eqnarray}
dE(t)^2={\delta}_{n,2}\,|z|^2\,dt
\end{eqnarray}

where ${\delta}_{n,2} $ is Kronecker's Delta, and (3.18), (3.20) now follow by substituting (3.22) and (3.23) in (3.21).

\end{proof}

\begin{theorem}
 The processes $B=\{B_t\,/\,t \geq 0\}$ and $P=\{P_t\,/\,t \geq 0\}$ defined by

\begin{eqnarray}
&B_t=A_t+A_t^{\dagger}&
\end{eqnarray}

and

\begin{eqnarray}
&P_t=\Lambda_t+\sqrt{\lambda}(A_t+A_t^{\dagger})+\lambda t&
\end{eqnarray}

are identified with Brownian motion and  Poisson process of intensity $\lambda$ respectively, in the sense that  their vacuum characteristic functionals are given by

\begin{eqnarray}
&<\psi (0),e^{i\,s\,B_t}\,\psi (0)>= e^{ - \frac{s^2}{ 2}\,t } &
\end{eqnarray}

and

\begin{eqnarray}
&<\psi (0),e^{i\,s\,P_t}\,\psi (0)>=e^{ \lambda\,\left(e^{i\,s}-1\right)\,t }&
\end{eqnarray} 

\end{theorem}

\begin{proof} Let

\begin{eqnarray}
U(t)=e^{i\,E(t)}
\end{eqnarray} 

where

\begin{eqnarray}
E(t)=s\,\left(A_t+A_t^{\dagger}\right)
\end{eqnarray} 

with $s\in \mathbb{R}$. Then, by Lemma 1

\begin{eqnarray}
dU(t)=U(t)\left[-\frac{s^2}{2}\,dt+i\,s\,dA_t+i\,s \,dA^{\dagger}_t \right]
\end{eqnarray}

or in integral form

\begin{eqnarray}
U(t)=I+\int_0^t\,U(r)\left[-\frac{s^2}{2}\,dr+i\,s\,dA_r+i\,
s \,dA^{\dagger}_r \right]
\end{eqnarray}

which, using Theorem 1 to compute vacuum expectations, yields

\begin{eqnarray}
<\psi (0),U(t)\,\psi (0)>=<\psi (0)\,\psi (0)>-\frac{s^2}{2}\int_0^t\,<\psi (0),U(r)\,\psi (0)>\,dr
\end{eqnarray}

Letting 

\begin{eqnarray}
f(t)=<\psi (0),U(t)\,\psi (0)>
\end{eqnarray}

we find that $f$ satisfies the differential equation

\begin{eqnarray}
f^{\prime}(t)=-\frac{s^2}{2}\,f(t)
\end{eqnarray}

with initial condition

\begin{eqnarray}
f(0)=1
\end{eqnarray}

and so

\begin{eqnarray}
f(t)=e^{ - \frac{s^2}{ 2}\,t }
\end{eqnarray}

thus proving (3.26). To prove (3.27) we let 

\begin{eqnarray}
U(t)=e^{i\,E(t)}
\end{eqnarray} 

where

\begin{eqnarray}
E(t)=s\,\left(\lambda \,t+\sqrt{\lambda}\,\left(A_t+A_t^{\dagger}\right)+{\Lambda}_t\right)
\end{eqnarray} 

with $s\in \mathbb{R}$. Then, as in the proof of (3.26)

\begin{eqnarray}
U(t)=I+\int_0^t\,U(r)\,\left(e^{i\,s}-1\right)  \left[ \lambda\,dr+ \sqrt{\lambda} \,\left(dA_r+dA^{\dagger}_r\right) +\,d{\Lambda}_r\right]
\end{eqnarray}

and

\begin{eqnarray}
<\psi (0),U(t)\,\psi (0)>=<\psi (0)\,\psi (0)>+\int_0^t\,\lambda\,\left(e^{i\,s}-1\right)<\psi (0),U(r)\,\psi (0)>\,dr
\end{eqnarray}

Letting 

\begin{eqnarray}
g(t)=<\psi (0),U(t)\,\psi (0)>
\end{eqnarray}

we find that $g$ satisfies the differential equation

\begin{eqnarray}
g^{\prime}(t)=\lambda\,\left(e^{i\,s}-1\right)\,g(t)
\end{eqnarray}

with initial condition

\begin{eqnarray}
g(0)=1
\end{eqnarray}

and so

\begin{eqnarray}
g(t)=e^{ \lambda\,\left(e^{i\,s}-1\right)\,t }
\end{eqnarray}

thus proving (3.27).

\end{proof}

Within the framework of Hudson-Parthasarathy Quantum Stochastic Calculus, classical  quantum mechanical evolution equations  take the form

\begin{eqnarray}
dU_t&=&-\left(\left(iH+\frac{1}{2}\,L^*L\right)\,dt+ L^* \,W\,dA_t -L\, dA_t^{\dagger}+\left(1-W\right)\,d\Lambda_t\right)U_t\\
U_0&=&1
\end{eqnarray}

where, for each $t\geq0$, $U_t$ is a unitary operator defined on the tensor product  $\mathcal{H} \otimes \Gamma(L^2(\mathbb{R}_+,\mathcal{ \mathcal{C}  }))$ of a system Hilbert space $\mathcal{H} $ and the noise (or reservoir) Fock space $\Gamma$.  Here $H$, $L$, $W$ are in $\mathcal{B}(\mathcal{H})$, the space of bounded linear operators on $\mathcal{H} $, with $W$ unitary and $H$ self-adjoint.  Notice that for $L=W=-1$ equation (3.45) reduces to a classical SDE of the form (3.2).  Here and in what follows we identify time-independent, bounded, system space operators $X$ with their ampliation $X \otimes 1$ to $\mathcal{H} \otimes \Gamma(L^2(\mathbb{R}_+,
\mathcal{\mathcal{C}}))$.

The quantum stochastic differential equation satisfied by the quantum flow

\begin{eqnarray}
&j_t(X)=U^*_t\,X\,U_t&
\end{eqnarray}

where $X$ is a bounded system space operator, is

\begin{eqnarray}
dj_t(X)&=&j_t\left( i\,\left[H,X\right]-\frac{1}{2}\left(L^*LX+XL^*L-2L^*XL\right)\right)\,dt\\
&+&j_t\left(\left[L^*,X\right]\,W\right)\,dA_t +j_t\left(W^*\,\left[X,L\right]\right)\,dA_t^{\dagger}+j_t\left(W^*\,X\,W-X\right)\,d\Lambda_t \nonumber\\
j_0(X)&=&X,\,\,\,\,t\in [0,T]
\end{eqnarray}

The commutation relations associated with the  operator processes $A_t$, $A^{\dagger}_t$  are the Canonical (or Heisenberg)  Commutation Relations (CCR), namely

\begin{eqnarray}
&\left[A_t,A^{\dagger}_t\right]=t\,I&
\end{eqnarray}

Classical and quantum stochastic calculi were unified by Accardi, Lu, and Volovich in \cite{7} within the framework of the white noise theory of T. Hida. Denoting the basic white noise functionals by $a_t$ and $a^{\dagger}_t$, they showed that the stochastic differentials of the Hudson-Partasarathy processes can be written as

\begin{eqnarray}
&dA_t=a_t\,dt& \\
&dA^{\dagger}_t = a^{\dagger}_t\,dt  &\\
&d \Lambda_t=  a_t a^{\dagger}_t\,dt &
\end{eqnarray}

and Hudson-Partasarathy stochastic differential equations are reduced to white noise equations. This unification started a whole new theory corresponding to quantum stochastic processes given by powers of the white noise functionals.

\section{Square of white noise quantum stochastic calculus} 

The square of white noise (SWN) commutation relations are a functional extension of the $sl(2;\mathbb{R})$ commutation relations

\begin{eqnarray}
&[B^-,B^+]=M,\,[M,B^+]=2B^+,\,[M,B^-]=-2B^-&
\end{eqnarray}

where

\begin{eqnarray}
&(B^-)^*=B^+,\,M^*=M.&
\end{eqnarray}

Following "renormalization", the SWN noise differentials are initially defined by

\begin{eqnarray}
&dB_t^-=a_t^2\,dt,\, dB_t^+={a_t^{\dagger}}^2\,dt,\, dM_t=a_t^{\dagger}\,a_t\,dt.&
\end{eqnarray}

A representation of the $sl(2;\mathbb{R})$ Lie algebra on $l_2(\mathbb{N})$ is defined by

\begin{eqnarray}
&\rho^+({B^+}^nM^k{B^-}^l)\,e_m=\theta_{n,k,l,m}\,e_{n+m-l}.&
\end{eqnarray}

where  {$e_m,\,m=0,1,2,\cdots$} is any orthonormal basis of $l_2(\mathbb{N})$, 

\begin{equation}
\theta_{n,k,l,m}:=
H(n+m-l)\,\sqrt{\frac{m-l+n+1}{m+1}}\,2^k(m-l+1)_{n}(m+1)^{(l)}(m-l+1)^{k},
\end{equation}

$H(x)$ is the Heaviside function $(H(x)=0$ for $x<0$; $H(x)=1$ for $x\geq0)$,

\begin{eqnarray*}
&0^0=1,\,\,\, (B^+)^n=(B^-)^n=N^n=0,\,\,\mbox{ for }n<0,&
\end{eqnarray*}

and   "factorial powers" are defined by 

\begin{eqnarray*}
&x^{(n)}=x(x-1)\cdots(x-n+1)&\\
&(x)_n=x(x+1)\cdots(x+n-1)&\\
&(x)_0=x^{(0)}=1.&
\end{eqnarray*}

Using this representation we obtain

\begin{eqnarray}
dM_t&=&d\Lambda_t(\rho^+(M))+dt\\
dB_t^+&=&d\Lambda_t(\rho^+(B^+))+dA^{\dagger}_t(e_0)\\
dB_t^-&=&d\Lambda_t(\rho^+(B^-))+dA_t(e_0).
\end{eqnarray}

To obtain a closed It\^{o} multiplication table we use as basic SWN  differentials

\begin{eqnarray}
d\Lambda_{n,k,l}(t) &=&d\Lambda_t(\rho^+({B^+}^nM^k{B^-}^l))\\
dA_m(t)&=&dA_t(e_m)\\
dA^{\dagger}_m(t)&=&dA^{\dagger} _t(e_m).
\end{eqnarray}

where $n,k,l,m\in\{0,1,...\}$, with It\^{o} multiplication  table 

\begin{eqnarray}
d\Lambda_{\alpha,\beta,\gamma}(t)\,d\Lambda_{a,b,c}(t)&=& \sum \,c_{\beta,\gamma,a,b}^{\lambda,\rho,\sigma,\omega,\epsilon} \, d\Lambda_{a+\alpha-\gamma+\lambda,\omega+\sigma+\epsilon,\lambda+c}(t)\\
d\Lambda_{\alpha,\beta,\gamma}(t)\,dA^{\dagger}_n(t)&=&\theta_{\alpha,\beta,\gamma,n}\,dA_{\alpha+n-\gamma}^{\dagger}(t)\\
dA_m(t)\,d\Lambda_{a,b,c}(t)&=&\theta_{c,b,a,m}\,dA_{c+m-a}(t)\\
dA_m(t)\,dA^{\dagger}_n(t)&=&\delta_{m,n}\,dt 
\end{eqnarray}
where

\begin{eqnarray}
&c_{\beta,\gamma,a,b}^{\lambda,\rho,\sigma,\omega,\epsilon}=\binom{\gamma}{\lambda}\binom{\gamma-\lambda}{\rho}\binom{\beta}{\omega}\binom{b}{\epsilon}2^{\beta+b-\omega-\epsilon}S_{\gamma-\lambda-\rho,\sigma}
 a^{(\gamma-\lambda)}(a+\lambda-1)^{(\rho)}(a-\gamma+\lambda)^{\beta-\omega}\lambda^{b-\epsilon},&\nonumber
\end{eqnarray}

$S_{\gamma-\lambda-\rho,\sigma}$ are the Stirling numbers of the first kind and  $\sum$  in (4.28) denotes the finite sum

\begin{eqnarray*}
\sum_{\lambda=0}^{\gamma} \sum_{\rho=0}^{\gamma-\lambda}\sum_{\sigma=0}^{\gamma-\lambda-\rho}\sum_{\omega=0}^{\beta}\sum_{\epsilon=0}^b.
\end{eqnarray*}

All other products of differentials are equal to zero. 

Square of white noise quantum evolutions are of the form

\begin{eqnarray}
&dU_t=((-\frac{1}{2}\,(D_-^*|D_-^*)+iH)\,dt+d\mathcal{A}_t(D_-)+d\mathcal{A}^{\dagger}_t(-r(W)D_-^*)+d\mathcal{L}_t(W-I))\,U_t&\\
&U_0=1&
\end{eqnarray}

while Langevin equations are of the form

\begin{eqnarray}
&dj_t(X)=j_t(i\,[X,H]-\frac{1}{2}\{(D_-^*|D_-^*)\,X\}+(r(W)D_-^*|X\,r(W)D_-^*))\,dt&\\
&+j_t(d\mathcal{A}^{\dagger}_t(D_-^*\,X-r(W^*\,X \circ W)D_-^*))+j_t(d\mathcal{A}_t(X\,D_-l(W^*\circ X\,W)D_-))+j_t(d\mathcal{L}_t(W^* \,X \circ W-X))&\nonumber\\
&j_0(X)=X,\,t\in [0,T].&
\end{eqnarray}

where $H$ is a  bounded self-adjoint system operator, $W$ is a $\circ$-product (see (4.28) for the definition of the $\circ$-product) unitary operator 
and $D_-=\sum_mD_{-,m}\otimes e_m $, where the $D_m$'s are bounded  system operators.

In equations (4.17)-(4.18) and (4.19)-(4.20) we have used 

\bigskip

(i) evolution coefficients:

\begin{eqnarray}
D_+&=&\sum_nD_{+,n}\otimes e_n\\
D_-&=&\sum_mD_{-,m}\otimes e_m\\
D_1&=& \sum_{\alpha,\beta,\gamma}D_{1,\alpha,\beta,\gamma} \otimes \rho^+({B^+}^{\alpha}M^{\beta}{B^-}^{\gamma})\\
E_1&=& \sum_{a,b,c}E_{1,a,b,c} \otimes \rho^+({B^+}^{a}M^{b}{B^-}^{c})
\end{eqnarray}
 
where the left hand sides of the tensor products corespond to bounded system operators

\bigskip

(ii) module operators $\mathcal{A}$, $\mathcal{A}^{\dagger}$ and $\mathcal{L}$ genericaly defined by:

\begin{eqnarray}
\mathcal{A}(a\otimes\xi)&=&a\otimes A(\xi)\\
\mathcal{A}^{\dagger}(a\otimes\xi)&=&a\otimes A^{\dagger}(\xi)\\
\mathcal{L}(a \otimes T)&=&a \otimes \Lambda(T)
\end{eqnarray}

\medskip
and
\medskip 

(iii) basic operations:

\begin{eqnarray}
D_1\circ E_1 &=&\sum_{\alpha,\beta,\gamma,a,b,c}\sum  c_{\beta,\gamma,a,b}^{\lambda,\rho,\sigma,\omega,\epsilon}D_{1,\alpha,\beta,\gamma}E_{1,a,b,c}\otimes \rho^+({B^+}^{a+\alpha-\gamma+\lambda}M^{\omega+\sigma+\epsilon}{B^-}^{\lambda+c})\\
(D_-^*|D_+)&=&
\sum_n \,D_{-,n}\,D_{+,n}\otimes 1\\
r(D_1)D_+&=&\sum_{n,\alpha,\beta,\gamma}\, D_{1,\alpha,\beta,\gamma}\theta_{\alpha,\beta,\gamma,n-\alpha+\gamma}D_{+,n-\alpha+\gamma}\otimes e_n\\
l(E_1)D_-&=&\sum_{n,\alpha,\beta,\gamma}\, D_{_,n+\alpha-\gamma}\theta_{\gamma,\beta,\alpha,n+\alpha-\gamma}E_{1,\alpha,\beta,\gamma}\otimes e_n
\end{eqnarray}

where $\sum$ is as in (4.16). The SWN Ito table can be concisely written as

\begin{eqnarray}
d\mathcal{A}_t(D_{-})\,d\mathcal{A}_t^{\dagger}(D_{+})&=&(D_{-}^*|D_{+})\,dt\\
d\mathcal{L}_t(D_1)\,d\mathcal{L}_t(E_1)&=&d\mathcal{L}_t(D_1\circ E_1)\\
d\mathcal{L}_t(D_1)\,d\mathcal{A}_t^{\dagger}(D_{+})&=&d\mathcal{A}_t^{\dagger}(r(D_1)D_{+})\\
d\mathcal{A}_t(D_{-})\, d\mathcal{L}_t(E_1)&=& d\mathcal{A}_t(l(E_1)D_{-}).
\end{eqnarray}

\section{Representation free quantum stochastic calculus} 

In quantum probability, starting with an operator representation of a Lie algebra, operator analogues and generalizations of classical stochastic noise processes such as Brownian motion, the exponential process, the Poisson process and others, as well as of purely quantum noises such as the square of white noise can be constructed as in section 3.

The quantum stochastic calculi constructed in order to study evolutions driven by these operator noises were dependent on the particular representation and led to analytic difficulties such as the unboundedness of  solutions of stochastic differential equations, the non-invariance of their domain e.t.c. These problems were removed by the introduction of the  representation free calculus of \cite{5} which provided the analytic and topological framework for a unified treatment of  quantum noises and which is now in standard use. In this section  we provide  a brief review of the  representation free calculus.

Let $H$ be a complex separable Hilbert space, $B(H)$ the algebra of all bounded linear operators on $H$, $D$ a total subset of $H$, $(A_{t]})_{t \geq 0}$ an increasing family of $W^*$-algebras of operators on $H$, $A$ a $W^*$-algebra of operators on $H$ such that $A_{t]} \subseteq A$ for all $t \geq 0$, $A_{t]}^{\prime}$ the commutant of  $A_{t]}$ in $B(H)$, $H_{t]}(\xi)$ the closure for each $\xi \in D$ of the subspace $[A_{t]}]=\{\alpha \xi \,/\, \alpha \in A_{t]}\}$, $e_{t]}^{\xi}$ the orthogonal projection onto $H_{t]}(\xi)$, $L(D;H)$ the vector space of all linear operators $F$ with domain containing $D$ and such that the domain of the adjoint operator $F^*$ also includes $D$.

A random variable is an element of  $L(D;H)$. A stochastic process in $H$ is a family $F=\{F_t\,/\,t \geq 0\}$ of random variables such that for each $\eta \in D$ the map $t \rightarrow F_t \eta$ is Borel measurable. If  $F_t \in B(H)$ for each $t \geq 0$, and  $\sup_{0 \leq t \leq T} \|F_t\| < + \infty$ for each $T < + \infty$, then the process $F$ is called locally bounded. If  $F_t \geq 0$ for each $t \geq 0$ then $F$ is positive.

Let $D_t^{\prime}$ denote the linear span of $A_{t]}^{\prime}D$. An operator $F$ is $t$-adapted to $A_{t]}$ if $dom(F)=D_t^{\prime} \subseteq dom(F^*)$ and, for all $\alpha_t^{\prime} \in  A_{t]}^{\prime}$ and $\xi \in D$,  $F\alpha_t^{\prime}\xi=\alpha_t^{\prime}F\xi$ and $F^*\alpha_t^{\prime}\xi=\alpha_t^{\prime}F^*\xi$. Strong limits of sequences of $t$-adapted operators are  $t$-adapted. A stochastic process $F$ is adapted to the filtration $\{A_{t]}\,/\,t \geq 0\}$ if $F_t$ is adapted for all $t \geq 0$, and it is simple if $F_t=\sum_{k=1}^n \chi_{[t_k,t_{k+1})}(t) F_{t_k}$ for some finite integer $n$ and $0 \leq t_0 < t_1 <...<t_{n+1} <+ \infty$.

An additive process is a family $M=\{M(s,t)\,/\,0 \leq s \leq t\}$ of random variables such that for all $s \leq t$ the operator $M(s,t)$ is $t$-adapted and, for all $r, s, t$ with $r \leq s \leq t$, $M(r,t)=M(r,s)+M(s,t)$ and $M(t,t)=0$ on $D$. To every additive process $M$ we associate the adapted process $M(t)=M(0,t)$ and conversely to every adapted process $\{M(t)\,/\, t \geq 0\}$ we associate the additive process $M(s,t)=M(t)-M(s)$. An additive process $M$ is regular if, for all $\xi \in D$ and  $r \leq s \leq t$, $H_{r]}(\xi) \subseteq dom(\overline {M^\#} (s,t))$ and $M^*(s,t)D \subseteq  D_s^{\prime}$, where $\overline{ M}$ denotes the closure of $M$ and $M^\#$ denotes either $M$ or $M^*$.  

If $M$ is a regular additive process and $F$ is a simple adapted process then the left (resp. right) stochastic integral  of $F$ with respect to $M$ over the interval $[0,t]$ is defined as an operator on $D_t^{\prime}$ by

\begin{eqnarray}
\int_0^t\,dM_s\,F_s=\sum_{k=1}^n\overline{ M} (t_k \wedge t,t_{k+1} \wedge t_0)\,F_{t_k}|_{D_t^{\prime}}
\end{eqnarray}

\begin{eqnarray}
\mbox{ (resp. }\int_0^t\,F_s \, dM_s=\sum_{k=1}^n F_{t_k}\,
M (t_k \wedge t,t_{k+1} \wedge t_0))
\end{eqnarray}

An additive regular process is an integrator of scalar type if for each $\xi \in D$ there exists a finite set $J(\xi) \subseteq D$ such that for each simple process $F$ and $t \geq 0$

\begin{eqnarray}
\|\int_0^t\,dM_s\,F_s \xi\|^2 \leq c_{t,\xi} \int_0^t \, d \mu_{\xi}(s)\,\sum_{\eta \in J(\xi)}\|F_s\eta\|^2
\end{eqnarray}

and

\begin{eqnarray}
\|\int_0^t\,F_s^* \,dM_s^*\xi\|^2 \leq c_{t,\xi} \int_0^t \, d \mu_{\xi}(s)\,\sum_{\eta \in J(\xi)}\|F_s^*\eta\|^2
\end{eqnarray}

and also, for all $\eta \in J(\xi)$, 

\begin{eqnarray}
J(\eta) \subseteq J(\xi)
\end{eqnarray}

where $c_{t,\xi} \geq 0$  and $\mu_{\xi}$ is a positive, locally finite, non atomic measure. 

If $M$ is an integrator of scalar type and, for all $\xi \in D$, $\mu_{\xi}$  is absolutely continuous with respect to Lebesgue measure then the stochastic integral with respect to $M$ can  be extended by continuity to processes $F \in L^2_{loc}([0,+\infty),dM)$, the space of all adapted processes $F$ with the topology induced by the seminorms

\begin{eqnarray}
\|F\|^2_{\eta,t,\mu_{\xi}}=\int_0^t\,\|F_s\eta\|^2\,d\mu_{\xi}(s)
\end{eqnarray}

such that for all $\xi \in D$, $\eta \in J(\xi)$ and $0 \leq t <+\infty$

\begin{eqnarray}
\int_0^t\,(\|F_s\eta\|^2+\|F^*_s\eta\|^2)\,d\mu_{\xi}(s) < +\infty
\end{eqnarray}

The thus extended stochastic integral has the usual linearity properties, and the maps $(s,t) \rightarrow \int_s^t\,dM_z\,F_z$ and $(s,t) \rightarrow \int_s^t\,F_z\,dM_z$  are additive, adapted processes, strongly continuous on $D$.

Suppose that $I$ is a set of finite cardinality  and let $\{M_{\alpha}\,/\,\alpha \in I\}$ be a set of integrators of scalar type.  Consider the quantum stochastic differential equation

\begin{eqnarray}
X(t)=X_0+\int_0^t\,\sum_{\alpha \in I}\,dM_{\alpha}(s)\,F_{\alpha}(s)X(s)G_{\alpha}(s)
\end{eqnarray}

or in differential form

\begin{eqnarray}
dX(t)=\sum_{\alpha \in I}\,dM_{\alpha}(t)\,F_{\alpha}(t)X(t)G_{\alpha}(t),\,\,\,X(0)=X_0,\,\,t \geq 0
\end{eqnarray}

where $X_0 \in A_{0]}$ and the coefficients $F_{\alpha},G_{\alpha}$ are locally bounded adapted processes leaving the domain $D$ invariant. If for all $\alpha \in I$, for all adapted processes $P$ integrable with respect to $M_{\alpha }$, and for all continuous functions $u,v$ on $[0,+\infty)$ satisfying $u(s) \leq s, v(s) \leq s$ for all $s \in [0,+\infty)$, the family of operators on $H$ $\{s \rightarrow F_{\alpha}(u(s))\,P(s)\,G_{\alpha}(v(s))\}$ is an adapted process integrable with respect to $M_{\alpha }$, then the above quantum stochastic differential equation has a unique locally bounded solution $X$ which is strongly continuous on $D$.

The above result can easily be extended to equations of the form

\begin{eqnarray}
dX(t)=\sum_{\alpha \in I}\,dM_{\alpha}(t)\,F_{\alpha}(t)(w(t)X(t)+z(t))G_{\alpha}(t),\,\,\,X(0)=X_0,\,\,t \geq 0
\end{eqnarray}

where $w,z$ are locally bounded adapted processes leaving the domain $D$ invariant.

 In what follows we restrict the term  process to processes leaving the domain $D$ invariant, and we denote the *-algebra of all processes by $W$.

If $M=\{M_{\alpha}\,/\,\alpha \in I \}$ is a self-adjoint family of regular integrator processes (i.e $M_{\alpha} \in M$ implies $(M_{\alpha})^* \in M$) then for all $s,t \in [0,+\infty)$ with $s < t$, for all $\alpha, \beta \in I$, for all $\xi \in D$, and for all adapted processes $F$, the Meyer Bracket or mutual quadratic variation  of $M_{\alpha}$ and $M_{\beta}$, defined by

\begin{eqnarray}
[[M_{\beta},M_{\alpha}]](s,t)=\lim_{|\Pi| \rightarrow 0}\, \sum\,M_{\beta}(t_{k-1},t_k)\,M_{\alpha}(t_{k-1},t_k)\,F_s\xi
\end{eqnarray}

where $\Pi$ is a partition of $[s,t]$,  exists in norm and defines an additive adapted process satisfying

\begin{eqnarray}
&M_{\beta}(s,t)\,M_{\alpha}(s,t)\,F_s\xi= \{\int_s^t\,dM_{\beta}(r)\,M_{\alpha}(s,r)+\int_s^t\,dM_{\alpha}(r)\,M_{\beta}(s,r)+[[M_{\beta},M_{\alpha}]](s,t)\}\,F_s\xi&
\end{eqnarray}

Assuming, for each pair $(\alpha,\beta) \in I \times I$, the existence of a family $\{c^{\gamma}_{\alpha \beta}\,/\,\gamma \in I \}$ of structure processes such that for each $s,t \in [0,+\infty)$ with $s < t$

\begin{eqnarray}
[[M_{\beta},M_{\alpha}]](s,t)=\sum_{\gamma \in I}\,\int_s^t\,c^{\gamma}_{\alpha \beta} (r)\,dM_{\gamma}(r)
\end{eqnarray}

and defining the differential of an additive process $M$ by

\begin{eqnarray}
dM(t)=M(t,t+dt)
\end{eqnarray}

we obtain

\begin{eqnarray}
d(M_{\beta}\,M_{\alpha})(t)=dM_{\beta}(t)\,M_{\alpha}(t)+M_{\beta}(t)\,dM_{\alpha}(t)+dM_{\beta}(t)\,dM_{\alpha}(t)
\end{eqnarray}

where the last product on the right is computed with the use of the  It\^{o} table

\begin{eqnarray}
dM_{\beta}(t)\,dM_{\alpha}(t)=\sum_{\gamma \in I}\,c^{\gamma}_{\alpha \beta} (t)\,dM_{\gamma}(t)
\end{eqnarray}

Assuming further that the $M_{\alpha}\,'s$ satisfy a  $\rho$-commutation relation i.e that for each $\alpha \in I$ there exists an automorphism $\rho_{\alpha}$ of $W$ mapping adapted processes into adapted processes, and such that

\begin{eqnarray}
 \rho_{\alpha}^2=id
\end{eqnarray}

\medskip

where $id$ denotes the identity map,  and for every $\xi \in D$, $s < t$, and adapted processes $F$, $F_s\xi \in dom(M_{\alpha}(s,t))$, $M_{\alpha}(s,t)\xi \in dom(\rho_{\alpha}(F_s))$,

\begin{eqnarray}
M_{\alpha}(s,t)F_s\xi=\rho_{\alpha}(F_s)M_{\alpha}(s,t)\xi  
\end{eqnarray}

and

\begin{eqnarray}
F_sM_{\alpha}(s,t)\xi=M_{\alpha}(s,t)\rho_{\alpha}(F_s)\xi
\end{eqnarray}

i.e stochastic processes commute with the stochastic differentials of the integrators,  we can extend (5.16) to processes $X=\{X(t)\,/\,t \geq 0\}$ and  $Y=\{Y(t)\,/\,t \geq 0\}$ of the form

\begin{eqnarray}
X(t)=\sum_{\alpha \in I}\,\int_s^t\,dM_{\alpha}(z)\,H_{\alpha}(z)\,,\,\,Y(t)=\sum_{\alpha \in I}\,\int_s^t\,dM_{\alpha}(z)\,K_{\alpha}(z)
\end{eqnarray}

where $H_{\alpha}, K_{\alpha}$ are for each  $\alpha \in I$ strongly continuous adapted processes. 

We thus have

\begin{eqnarray}
d(X \, Y)(t)=dX(t)\,Y(t)+X(t)\,dY(t)+dX(t)\,dY(t)
\end{eqnarray}

where

\begin{eqnarray}
dX(t)=\sum_{\alpha \in I}\,dM_{\alpha}(t)\,H_{\alpha}(t),\,\,dY(t)=\sum_{\alpha\in I}\,dM_{\alpha}(t)\,K_{\alpha}(t)
\end{eqnarray}
 
$dX(t)\,dY(t)$ is computed with the use of the It\^{o} table (5.16), and (5.21) is understood weakly on $D$ i.e for all $\xi,\eta \in D$ 

\begin{eqnarray}
<d(X \, Y)(t)\xi,\eta>=<[dX(t)\,Y(t)+X(t)\,dY(t)+dX(t)\,dY(t)]\xi,\eta>
\end{eqnarray}

\section{Quantum stochastic differential equations from white noise equations and quantum mechanics} 

In the Schroedinger picture of Quantum Mechanics the initial state of a quantum system is described by a ket-vector $|{\psi}(0)>$ evolving under the influence of a Hamiltonian operator $H=H_0+H_I$, where $H_0$ and $H_I$ are the "free" and "interaction" parts, respectively, to a state given at time $t$ by a ket-vector  $|{\psi}(t)>$ satisfying the Schroedinger equation

\begin{eqnarray}
i\,\hbar\,d|\psi (t)>=H\,|\psi (t) >\,dt 
\end{eqnarray}

or in operator form

\begin{eqnarray}
|{\psi}(t)>=V_t \,|{\psi}(0)>
\end{eqnarray}

where the unitary evolution operator $V_t$ is given by

\begin{eqnarray}
V_t=e^{-\frac{i}{ \hbar}\,t\,H}.
\end{eqnarray}

In the Heisenberg picture of Quantum Mechanics, it is the observables $X$, i.e the  self-adjoint operators on the system space that vary with time generating a flow

\begin{eqnarray}
X_t=V^*_t\,X\,V_t
\end{eqnarray}

where

\begin{eqnarray}
V_t=e^{-\frac{i}{ \hbar}\,t\,H}
\end{eqnarray}

and

\begin{eqnarray}
dX_t=\frac{i}{ \hbar}\,[H,X_t]\,dt,\,\,\,\,\,\,\, X_0=X
\end{eqnarray}

In  Interaction Representation one has

\begin{eqnarray}
|{\psi}_I(t)>=e^{\frac{i}{ \hbar}\,t\,H_0}\,|\psi (t)>=e^{\frac{i}{ \hbar}\,t\,H_0}\,e^{-\frac{i}{ \hbar}\,t\,H}\,|\psi (0)>=U_t\,|\psi (0)>
\end{eqnarray}

where the propagator

\begin{eqnarray}
U_t=e^{\frac{i}{ \hbar}\,t\,H_0}\,e^{-\frac{i}{ \hbar}\,t\,H}=e^{\frac{i}{ \hbar}\,t\,H_0}\,V_t
\end{eqnarray}

satisfies

\begin{eqnarray}
dU_t=-i\,H_I(t)\,U_t\,dt,\,\,\,\,\,\,  U_0=1
\end{eqnarray}

where

\begin{eqnarray}
H_I(t)=\frac{1}{ \hbar}\,e^{\frac{i}{ \hbar}\,t\,H_0}\,H_I\,e^{-\frac{i}{ \hbar}\,t\,H_0}
\end{eqnarray}

In the Stochastic Limit approach, developed by L.  Accardi, Y. G. Lu, and I. Volovich in \cite{6}, starting with a Hamiltonian  
$H=H(\lambda)$, where $\lambda \in \mathbb{R}$ is small (e.g $H=H_0+\lambda\,H_I$, $\lambda$: coupling constant), equation (6.9) takes the form

\begin{eqnarray}
dU_t^{(\lambda)}=-i\,H_I^{(\lambda)}(t)\,U_t^{(\lambda)}\,dt,\,\,\,\,U_0^{(\lambda)}=1
\end{eqnarray}

where $t \rightarrow +\infty$ sescribes the long-term behavior, (scattering theory), and $\lambda \rightarrow 0$ describes the weak effects e.g weak coupling, low density (perturbation theory). The time rescaling  $t \rightarrow \frac{t}{{\lambda}^2}$ puts things together by considering the long term cummulative effects of weak actions. The solution of the rescaled Hamiltonian equation (6.11)

\begin{eqnarray}
dU_{ \frac{ t }{ {\lambda }^2 } }^{ (\lambda ) }=-\frac{i}{ {\lambda }^2} \,H_I^{ (\lambda ) }(\frac{t}{  {\lambda }^2} )\,U_{\frac{t}{{\lambda }^2} }^{ (\lambda ) }\,dt,\,\,\,\,\,\,U_0^{ (\lambda ) }=1
\end{eqnarray}

converges, in the sense of matrix elements of collective vectors of \cite{6},   as $\lambda \rightarrow 0$ to a process $U_t$ satisfying

\begin{eqnarray}
dU_t=-i\,H_t\,U_t\,dt, \,\,\,\,\,\,U_0=1
\end{eqnarray}

whre $H_t$ is a singular Hamiltonian which can be written in terms of the Hida white noise functionals $b_t$ and $b_t^{\dagger}$ as

\begin{eqnarray}
H_t&=&\int_0^t\,\left(C_0(s)+\sum_{n,k}\,C_{n,k}(s){b_t^{\dagger}}^n\,b_t^k\right)\,ds=\int_0^t\,C_0(s)\,ds+\sum_{n,k}\,C_{n,k}(s)\,dB^n_k(s)
\end{eqnarray}

where we have used the noise differentials notation 

\begin{eqnarray}
dB^n_k(t)={b_t^{\dagger}}^n\,b_t^k\,dt
\end{eqnarray}

For $n,k\in\{0,1\}$ equation (6.13) reduces to a  quantum stochastic differential equation of the type discussed in section 3.

\section{Quadratic cost control }

\subsection{Control of first order white noise Langevin flows}

 In the notation of sections 3 and 4, we consider a quantum flow $\{j_t(X)/\,t\in [0,T]\}$ of bounded linear operators on  $\mathcal{H} \otimes \Gamma$ defined by $j_t(X)=U_t^*\,X \,U_t$ where $\mathcal{H}$ is a  separable Hilbert space, $\Gamma$ is the Boson Fock space over $L^2(\mathbb{R}_+,\mathbb{C})$, $X$ is a self-adjoint operator on $\mathcal{H}$ identified with its ampliation $X \otimes 1$ to $\mathcal{H} \otimes \Gamma$, and $U=\{U_t \, / \, t \geq 0 \}$ is a unitary evolution satisfying on  $\mathcal{H} \otimes \Gamma$ a quantum stochastic differential equation of the form

\begin{equation}
dU_t=-((iH+\frac{1}{2}\,L^*L)\,dt+ L^* \,W\,dA_t -L\, dA_t^{\dagger}+(1-W)\,d\Lambda_t)\,U_t,\,t\in [0,T]
\end{equation}
 with adjoint

\begin{equation}
dU_t^*=-U_t^*\,((-iH+\frac{1}{2}\,L^*L)\,dt- L^* dA_t +W^*\,L\, dA_t^{\dagger}+(1-W^*)\,d\Lambda_t),\,t \in [0,T]
\end{equation}

and initial conditions

\begin{equation}
U_0=U_0^*=1
\end{equation}

where $H,\,L,\,W$ are bounded operators on $\mathcal{H}$ with $H$ self-adjoint and $W$ unitary. These conditions guarantee the existence uniqueness and unitarity of the solution of (7.1), (7.2).

Using the It\^{o} table for first order white noise we can show  that the flow $\{j_t(X)/\,t\in [0,T]\}$ satisfies the  quantum stochastic differential equation 

\begin{eqnarray}
dj_t(X)&=&j_t(i[H,X]-\frac{1}{2}(L^*LX+XL^*L-2L^*XL))\,dt+j_t([L^*,X]\,W)\,dA_t \\
&+&j_t(W^*\,[X,L])\,dA_t^{\dagger}+ +j_t(W^*\,X\,W-X)\,d\Lambda_t\nonumber
\end{eqnarray}

with initial condition

\begin{eqnarray}
j_0(X)=X,\,t\in [0,T]
\end{eqnarray}

\begin{definition}
On a finite time interval $[0,T]$, the cost functional for the solution of the quantum Langevin equation (7.3) is given by: 
\begin{equation}
J_{\xi,T}(L,W)=\int_0^T\,[\,\|j_t(X)\,\xi\|^2+\frac{1}{4}\|j_t(L^*L)\,\xi\|^2\, ]\,dt+\frac{1}{2}\|j_T(L)\,\xi\|^2
\end{equation}
where  $\xi$ is an arbitrary vector in the exponential domain of $\mathcal{H}\otimes \Gamma$. 
\end{definition}

  The  functional  $J_{\xi,T}(\cdot)$ of Definition 4, which we propose as suitable for the evaluation of the performance of a quantum flow,  is derived from a quantum extension of the classical quadratic performance criterion for operator processes $X=\{X(t)\,/\,t \geq 0\}$ that are solutions of quantum stochastic differential equations such as (5.9), in the case when $X(t)=U(t)$ is a unitary operator for each $t \geq 0$. The operators $X(t)$ are not necessarily self-adjoint so they do not in general correspond to quantum mechanical observables in the classical sense unless special assumptions are made on the coefficients of the defining quantum stochastic differential equations. However, this is not a problem since the conrol problem for the $X(t)\,'s$ is to be used as a passage to the solution of the control problem for the corresponding quantum flow  $j_t(\cdot )$ which does consist of quantum mechanical observables. 

 Thinking of $L$ and $W$ as  controls we interpret the first term of the right hand side of (7.6) as a measure of the size of the flow over $[0,T]$, the second as a measure of the control effort over $[0,T]$ and the third as a "penalty" for allowing the evolution to go on for a long time. We consider the problem of controlling the size of such a flow by minimizing the cost functional $ J_{\xi,T}(L,W)$ of (7.6). 

\begin{theorem}
Let $U=\{U_t\,/\,t\geq 0\}$ be a  process satisfying the  quantum stochastic differential equation 

\begin{equation}
dU_t=(F\,U_t+u_t)\,dt+ \Psi \,U_t\, dA_t+ \Phi \,U_t\, dA_t^{\dagger}+Z\,U_t\,d\Lambda_t,\,U_0=1,\,t\in [0,T]
\end{equation}
with adjoint

\begin{equation}
dU_t^*=(U_t^*\,F^*+u_t^*)\,dt+ U_t^*\,\Psi^* \,dA_t^{\dagger}+ U_t^*\,\Phi^* dA_t+U_t^*\,Z^*\,d\Lambda_t,\,U_0^*=1,\,t\in [0,T]
\end{equation}

where $T > 0$ is a fixed finite horizon,  the coefficients $F,\,\Psi,\,\Phi, \,Z$  are bounded operators on the system space $\mathcal{H}$ and $u_t$ is of the form $-\Pi \,U_t$ for some positive bounded system operator $\Pi$.

Then the functional  

\begin{eqnarray}
 Q_{\xi,T}(u)=\int_0^T\,[<U_t \,\xi,X^2\,U_t\, \xi>+<u_t\, \xi,u_t\, \xi>]\,dt-<u_T\, \xi,U_T \,\xi>
\end{eqnarray}

where $X$ is a system space observable, identified with its ampliation 
$X \otimes I$ to $\mathcal{H}  \otimes \Gamma$, is minimized over the set of feedback control processes of the form $u_t=-\Pi \,U_t$, by choosing $\Pi$ to be a bounded,  positive, self-adjoint   system operator satisfying

\begin{eqnarray}
&\Pi \, F+F^* \Pi+{\Phi}^* \Pi \Phi-{\Pi}^2+X^2=0 &\\
&\Pi \,\Psi + {\Phi}^* \, \Pi + {\Phi}^* \, \Pi \, Z=0&\\
&\Pi \,Z +Z^* \,\Pi + Z^* \, \Pi \,Z=0.&
\end{eqnarray}

The minimum value is $<\xi,\Pi \xi>$.  We recognize (7.10) as the algebraic Riccati equation.

\end{theorem}

\begin{proof} 

Let 

\begin{eqnarray}
&\theta_t=<\xi,U_t^*\,\Pi \,U_t\,\xi>.&
\end{eqnarray}

Using the identity $d(xy)=x\,dy+dx\,y+dx\,dy$ we obtain

\begin{eqnarray}
&d\theta_t=<\xi,d(U_t^* \,\Pi \,U_t)\,\xi>=<\xi,(dU^*_t \, \Pi \,U_t+U_t^* \,\Pi \,dU_t+dU_t^* \,\Pi \,dU_t)\,\xi>& 
\end{eqnarray}

which, after replacing $dU_t$ and $dU_t^*$ by (7.7) and (7.8) respectively  and using the It\^{o} table of section 3,  becomes

\begin{eqnarray}
&d\theta_t=<\xi, U_t^*\,((F^*\,\Pi+\Pi\,F+\Phi^*\,\Pi\,\Phi)\,dt+(\Phi^*\,\Pi+\Pi\,\Psi+\Phi^*\,\Pi\,Z)\,dA_t&\\
&+(\Psi\,\Pi^*+\Pi\,\Phi+Z^*\,\Pi\,\Phi)\,dA_t^{\dagger}+(Z^*\,\Pi+\Pi\,Z+Z^*\,\Pi\,Z)\,d\Lambda_t)\,U_t\,\xi>&\nonumber\\
&+<\xi,(u_t^*\,\Pi\,U_t+U_t^*\,\Pi\,u_t)\,dt\,\xi>.&\nonumber
\end{eqnarray}

and by (7.10)-(7.12)

\begin{eqnarray}
&d\theta_t=<\xi, U_t^*\,( \Pi^2-X^2)\,U_t\,dt\,\xi>   +<\xi,(u_t^*\,\Pi\,U_t+U_t^*\,\Pi\,u_t)\,dt\,\xi>.&
\end{eqnarray}

By (7.11)

\begin{eqnarray}
&\theta_T-\theta_0=<\xi, U_T^*\,\Pi\,U_T\,\xi>-<\xi,\Pi\,\xi>.&
\end{eqnarray}

while by (7.14)

\begin{eqnarray}
&\theta_T-\theta_0=\int_0^T\,(<\xi, U_t^*\,( \Pi^2-X^2)\,U_t\,\xi>   +<\xi,(u_t^*\,\Pi\,U_t+U_t^*\,\Pi\,u_t)\,\xi>)\,dt.&
\end{eqnarray}

By (7.15) and (7.16)

\begin{eqnarray}
&<\xi, U_T^*\,\Pi\,U_T\,\xi>-<\xi,\Pi\,\xi>=\int_0^T\,(<\xi, U_t^*\,( \Pi^2-X^2)\,U_t\,\xi>   +<\xi,(u_t^*\,\Pi\,U_t+U_t^*\,\Pi\,u_t)\,\xi>)\,dt.&
\end{eqnarray}

Thus

\begin{eqnarray}
&Q_{\xi,T}(u)=(<\xi, U_T^*\,\Pi\,U_T\,\xi>-<\xi,\Pi\,\xi>)+Q_{\xi,T}(u) (<\xi, U_T^*\,\Pi\,U_T\,\xi>-<\xi,\Pi\,\xi>).&
\end{eqnarray}

Replacing the first parenthesis on the right hand side of (7.18) by (7.17), and $Q_{\xi,T}(u)$ by (7.9) we obtain after cancellations

\begin{eqnarray}
Q_{\xi,T}(u)&=&\int_0^T \,(<\xi,(U_t^* \,{\Pi}^2\,U_t+u_t^* \,\Pi \,U_t+U_t^* \,\Pi \,u_t+u_t^* \,u_t)\,\xi>\,dt+<\xi,\Pi\,\xi>\\
&=&\int_0^T \,||(u_t+\Pi\,U_t)\,\xi||^2\,dt+<\xi,\Pi\,\xi>\nonumber
\end{eqnarray}

which is clearly minimized by $u_t=-\Pi\,U_t$ and the minimum is $<\xi,\Pi\,\xi>$.

\end{proof}

\begin{definition}The pair ($i\,H$, $X$) is called stabilizable if there exists a bounded
system operator $K$ such that $i\,H+KX$ is the generator of an
asymptotically stable semigroup $\mathcal{F}_t$ i.e there exist constants $M>0$ and $\omega <0$ such that $||\mathcal{F}_t||\leq M\,e^{\omega\,t}$.
\end{definition}

\begin{theorem}
Let $X$ be a bounded self-adjoint system operator such that the pair ($i\,H$, $X$) is stabilizable. The quadratic performance functional (7.6) associated with the quantum stochastic flow $\{j_t(X)=U_t^*\,X \,U_t\,/\,t \geq 0\}$,  where $U=\{U_t\,/\,t\geq 0\}$ is the solution of (7.1), is minimized by 

\begin{eqnarray}
&L=\sqrt{2}\,\Pi^{1/2}\,W_1\,\,\,(\mbox{polar decomposition of}\,\,L)&  
\end{eqnarray}
and 

\begin{eqnarray}
&W=W_2&
\end{eqnarray}
where  $\Pi$ is a positive self-adjoint solution of the ``algebraic Riccati equation''
\begin{eqnarray}
&i\,[H,\Pi]+\Pi^2+X^2=0&  
\end{eqnarray}
and $W_1$, $W_2$ are bounded unitary system operators  commuting with  $\Pi$. Moreover 

\begin{eqnarray}
&\min_{L,W}\, J_{\xi,T}(L,W)= <\xi,\Pi\,\xi>&
\end{eqnarray}
independently of $T$.
\end{theorem}

\begin{remark}
Equation (7.24) is a special case of the algebraic Riccati equation (ARE).
If the pair ($i\,H$, $X$) is stabilizable, then (7.24) has a positive self-adjoint solution $\Pi$.
\end{remark}

\begin{proof}

Looking at (7.1) as (7.7) with $u_t=-\frac{1}{2}\,L^*\,L\,U_t$, $F=-i\,H$, $\Psi=-L^*\,W$, $\Phi=L$, and $Z=W-1$,   (7.6) is identical to (7.9). Moreover, equations  (7.10)-(7.12) become

\begin{eqnarray}
&i\,[H,\Pi]+L^*\,\Pi\,L-\Pi^2+X^2=0&\\
&L^*\,\Pi-\Pi\,L^*\,W+L^*\,\Pi\,(W-1)=0&\\
&(W^*-1)\,\Pi+\Pi\,(W-1)+(W^*-1)\,\Pi\,(W-1)=0.&
\end{eqnarray}

By the self-adjointness of $\Pi$, (7.27) implies that
\begin{eqnarray}
&[L,\Pi]=[L^*,\Pi]=0&
\end{eqnarray}
while (7.28) implies that
\begin{eqnarray}
&[W,\Pi]=[W^*,\Pi]=0&
\end{eqnarray}

i.e (7.23).  By (7.29) and the fact that in this case

\begin{eqnarray}
&\Pi=\frac{1}{2}\,L^*\,L \mbox{  i.e  }L^*\,L=2\,\Pi&
\end{eqnarray}

equation (7.26) implies (7.24). Equations (7.29) and (7.31) also imply that 

\begin{equation}
[L,L^*]=0\,\,(\mbox{i.e $L$ is normal})  
\end{equation}
which implies (7.22). 
\end{proof}

\subsection{Control of square of white noise Langevin flows  }

In the square of white noise case, equations (7.1) and (7.2) are replaced, respectively,  by 

\begin{eqnarray}
&dU_t=((-\frac{1}{2}\,(D_-^*|D_-^*)+iH)\,dt+d\mathcal{A}_t(D_-)+d\mathcal{A}^{\dagger}_t(-r(W)D_-^*)+d\mathcal{L}_t(W-I))\,U_t&
\end{eqnarray}

and

\begin{eqnarray}
&dU_t^*=U_t^*\,((-\frac{1}{2}\,(D_-^*|D_-^*)-iH)\,dt+d\mathcal{A}^{\dagger}_t(D_-^*)+d\mathcal{A}_t(-l(W^*)D_-)+d\mathcal{L}_t(W^*-I))&
\end{eqnarray}

with  initial conditions

\begin{eqnarray}
&U_0=U_0^*=1&
\end{eqnarray}

 where $H$ is any bounded self-adjoint system operator, $W$ is a $\circ$-product unitary operator such that $r(W)r(W^*)=r(W^*)r(W)=1$, $I$ is the $\circ$-product identity, $D_-$ is an arbitrary  operator , and $1$ is the identity operator on $\mathcal{H} \otimes \Gamma$. These conditions guarantee the existence uniqueness and unitarity of the solutions.

\begin{proposition}

In the case of the square of white noise, the quantum Langevin equation (7.4) is replaced by

\begin{eqnarray}
&dj_t(X)=j_t(i\,[X,H]-\frac{1}{2}((D_-^*|D_-^*)\,X+X\,(D_-^*|D_-^*))+(r(W)D_-^*|X\,r(W)D_-^*))\,dt&\\
&+j_t(d\mathcal{A}^{\dagger}_t(D_-^*\,X-r(W^*\,X)r(W)D_-^*))+j_t(d\mathcal{A}_t(X\,D_--l(X\,W)l(W^*)D_-))+j_t(d\mathcal{L}_t(W^* \,X \circ W-X))&\nonumber
\end{eqnarray}

\bigskip

with $j_0(X)=X,\,t\in [0,T]$.
\end{proposition}

\begin{proof}

\[dj_t(X)=(dU_t^*)\,X \,U_t+U_t^*\,X \,(dU_t)+(dU_t^*)\,X \,(dU_t)\]
\[=U_t^*\{(-\frac{1}{2}(D_-^*|D_-^*)-i\,H)\,X\,dt+d\mathcal{A}^{\dagger}_t(D_-^*\,X)-d\mathcal{A}_t(l(W^*)D_-\,X)\]
\[+d\mathcal{L}_t((W^*-I)\,X))+X\,(-\frac{1}{2}(D_-^*|D_-^*)+i\,H)\,dt-d\mathcal{A}^{\dagger}_t(X\,r(W)D_-^*)\]
\[+d\mathcal{A}_t(X\,D_-)+d\mathcal{L}_t(X\,(W-I))+(r(W)D_-^*|X\,r(W)D_-^*)\,dt\]
\[-d\mathcal{A}_t(X\,l((W-I))\,l(W^*)D_-)-d\mathcal{A}^{\dagger}_t(r((W^*-I)\,X)\,r(W)D_-^*)\]
\[+d\mathcal{L}_t((W^*-I)\,X\circ (W-I))\}\,U_t\]
\[=U_t^*\{(-\frac{1}{2}((D_-^*|D_-^*)\,X+X\,(D_-^*|D_-^*))+i\,[X,H]+\]
\[(r(W)D_-^*|X\,r(W)D_-^*))\,dt+d\mathcal{A}^{\dagger}_t(D_-^*\,X-X\,r(W)D_-^*\]
\[-r((W^*-I)\,X)\,r(W)D_-^*)-d\mathcal{A}_t(l(W^*)D_-\,X -X\,D_-+\]
\[ l(X\,(W-I))\,l(W^*)D_- )+d\mathcal{L}_t((W^*-I)\,X+X\,(W-I)\]
\[+(W^*-I)\,X\circ (W-I))\}\,U_t\]
\[=U_t^*\{(-\frac{1}{2}((D_-^*|D_-^*)\,X+X\,(D_-^*|D_-^*))+i\,[X,H]+\]
\[(r(W)D_-^*|X\,r(W)D_-^*)\}\,U_t\,dt\]
\[+U_t^*\{d\mathcal{A}^{\dagger}_t(D_-^*\,X-X\,r(W)D_-^*-r((W^*-I)\,X)\,r(W)D_-^*  )\}\,U_t\]
\[+U_t^*\{d\mathcal{A}_t(-l(W^*)D_-\,X +X\,D_--l(X\,(W-I))\,l(W^*)D_- )\}\,U_t\]               
\[ +U_t^*\{d\mathcal{L}_t(W^*\,X\circ W-X)\}\,U_t\]
\[=j_t(i\,[X,H]-\frac{1}{2}((D_-^*|D_-^*)\,X+X\,(D_-^*|D_-^*))+(r(W)D_-^*|X\,r(W)D_-^*))\,dt\]
\[+j_t(d\mathcal{A}^{\dagger}_t(D_-^*\,X-X\,r(W)D_-^*-r((W^*-I)\,X)\,r(W)D_-^* ))\]
\[+j_t(d\mathcal{A}_t( X\,D_--l(W^*)D_-\,X- l(X\,(W-I))\,l(W^*)D_-   ))\]
\[+ j_t(d\mathcal{L}_t(W^*\,X\circ W-X))\]
\[=j_t(i\,[X,H]-\frac{1}{2}((D_-^*|D_-^*)\,X+X\,(D_-^*|D_-^*))+\]
\[(r(W)D_-^*|X\,r(W)D_-^*))\,dt+\]
\[j_t(d\mathcal{A}^{\dagger}_t(D_-^*\,X-r(W^*\,X)r(W)D_))+j_t(d\mathcal{A}_t(X\,D_--l(X\,W)l(W^*)D_-))\]
\[+j_t(d\mathcal{L}_t(W^* \,X \circ W-X))\]
\end{proof}

\begin{definition}
On a finite time interval $[0,T]$, the cost functional for the solution of the quantum Langevin equation (7.36) is given by:
 
\begin{equation}
J_{\xi,T}(D_-,W)=\int_0^T\,[\,\|j_t(X)\,\xi\|^2+\frac{1}{4}\|j_t((D_-^*|D_-^*))\,\xi\|^2\, ]\,dt+\frac{1}{2}<\xi,j_T((D_-^*|D_-^*))\,\xi>
\end{equation}
where  $\xi$ is an arbitrary vector in the exponential domain of $\mathcal{H}\otimes \Gamma$. 
\end{definition}

The square of white noise analogues of Theorems 6 and 7 are as follows.

\begin{theorem}
Let $U=\{U_t\,/\,t\geq 0\}$ be a  process satisfying the  quantum stochastic differential equation 

\begin{equation}
dU_t=(F\,U_t+u_t)\,dt+  d\mathcal{A}_t(\Psi) \,U_t+ d\mathcal{A}^{\dagger}_t(\Phi) \,U_t +d\mathcal{L}_t(Z)\,U_t,\,U_0=1,\,t\in [0,T]
\end{equation}

with adjoint

\begin{equation}
dU_t^*=(U_t^*\,F^*+u_t^*)\,dt+U_t^*\, d\mathcal{A}_t^{\dagger}(\Psi^*)+\,U_t^*\, d\mathcal{A}_t(\Phi^*) + U_t^* \, d\mathcal{L}_t(Z^*),\,U_0^*=1,\,t\in [0,T]
\end{equation}

where $T > 0$ is a fixed finite horizon,  $F$ is a bounded operator on the system space $\mathcal{H}$,  $\Psi$, $\Phi$, and $Z$  are of the same form as $D_-$, $D_+$, and $D_1$ respectively, and $u_t$ is of the form $-\Pi \,U_t$ for some positive bounded system operator $\Pi$.

The functional  

\begin{eqnarray}
Q_{\xi,T}(u)=\int_0^T\,[<U_t \,\xi,X^2\,U_t\, \xi>+<u_t\, \xi,u_t\, \xi>]\,dt-<u_T\, \xi,U_T \,\xi>
\end{eqnarray}

where $X$ is a system space observable, identified with its ampliation 
$X \otimes I$ to $\mathcal{H}  \otimes \Gamma$, is minimized over the set of feedback control processes of the form $u_t=-\Pi \,U_t$ by choosing $\Pi$ to be a bounded,  positive, self-adjoint   system operator satisfying

\begin{eqnarray}
&\Pi \, F+F^* \Pi+(\Phi|\Pi\,\Phi)-{\Pi}^2+X^2=0 &\\
&\Pi \,\Psi + {\Phi}^* \, \Pi +l(\Pi \, Z){\Phi}^* \, =0&\\
&\Pi \,Z +Z^* \,\Pi + (Z^* \, \Pi)\circ \,Z=0.&
\end{eqnarray}

The minimum value is $<\xi,\Pi \xi>$. 

\end{theorem}

\begin{proof} The proof follows in a way  similar to that of Theorem 6 with the use of the square of white noise It\^{o} table (4.31)-(4.34) of section 4.
\end{proof}

\begin{theorem}  Let $X$ be a bounded self-adjoint system operator such that the pair ($i\,H$, $X$) is stabilizable. The quadratic performance functional (7.37) associated with the quantum stochastic flow $\{j_t(X)=U_t^*\,X \,U_t\,/\,t \geq 0\}$,  where  $U=\{U_t\,/\,t\geq 0\}$ is the solution of (7.33), is minimized by choosing 

\begin{eqnarray}
&D_-=\sum_n\,D_{-,n}\otimes e_n &
\end{eqnarray}

and 

\begin{eqnarray}
&W=\sum_{\alpha,\beta,\gamma}\,W_{\alpha,\beta,\gamma} \otimes \rho^+({B^+}^{\alpha}M^{\beta}{B^-}^{\gamma})&
\end{eqnarray}

such that

\begin{eqnarray}
&\frac{1}{2}\,(D_-^*|D_-^*)=(\frac{1}{2}\,\sum_n\,D_{-,n}\,D_{-,n}^*)\,\otimes \,1=\Pi,&  
\end{eqnarray}

and

\begin{eqnarray}
&[D_{-,n},D_{-,m}]=[D_{-,n},D_{-,m}^*]=0&\\
&[D_{-,n},W_{\alpha,\beta,\gamma}]=[D_{-,n},W_{\alpha,\beta,\gamma}^*]=0&
\end{eqnarray}

for all $n,m,\alpha,\beta,\gamma$, which also implies that $[D_{-,n}^*,W_{\alpha,\beta,\gamma}]=[D_{-,n}^*,W_{\alpha,\beta,\gamma}^*]=0$, where  $\Pi$ is a positive self-adjoint solution of the algebraic Riccati equation 

\begin{eqnarray}
&i\,[H,\Pi]+\Pi^2+X^2=0.&  
\end{eqnarray} 

Moreover 

\begin{eqnarray}
&\min_{D_{-},W}\, J_{\xi,T}(D_-,W)= <\xi,\Pi\,\xi>&
\end{eqnarray}

independently of $T$.

\end{theorem}

\begin{proof}
Looking at (7.33) as (7.38) with $u_t=-\frac{1}{2}\,(D_-^*|D_-^*)\,U_t$ i.e 
$\Pi=\frac{1}{2}\,(D_-^*|D_-^*)=(\frac{1}{2}\,\sum_n\,D_{-,n}\,D_{-,n}^*)\,\otimes \,1\ ,\ F=i\,H\ ,\ \Psi=D_-\ ,\ \Phi=-r(W)D_-^*$, and $Z=W-I$,   (7.40) is identical to (7.37) and equations  (7.41)-(7.43) become 

\begin{eqnarray}
&i\,[\Pi,H]+  (r(W)D_-^*|\Pi\,r(W)D_-^*)-\Pi^2+X^2=0&\\
&\Pi\,D_--l(W^*)D_-\,\Pi-l(\Pi\,(W-I))\,l(W^*)D_-=0&\\
&\Pi\,(W-I)+(W^*-I)\,\Pi+((W^*-I)\,\Pi)\circ (W-I)=0.&
\end{eqnarray}

Equation (7.53) implies $W^* \,\Pi \circ W=\Pi$ $\Rightarrow$ $W^*\,\Pi \circ W \circ W^*=\Pi \circ W^*$ $\Rightarrow$ $W^* \,\Pi \circ I=\Pi \circ W^*$ $\Rightarrow$ $W^*\,\Pi=\Pi \,W^*$ $\Rightarrow$ $[\Pi,W]=[\Pi,W^*]=0$ and (7.48) follows from (7.46). Similarly, (7.52) implies that $[\Pi,D_{-}]=0$ from which (7.47) follows. Finally, using the fact that $(r(W)D_{-}^*\,|\,r(W)D_{-}^*)=(D_{-}^*\,|\,r(W^*)r(W)D_{-}^*)=(D_{-}^*\,|\,D_{-}^*)$, (7.51) implies (7.49).

\end{proof}

 \subsection{ Representation free feedback control }

Within the framework of the representation free calculus described in  section 5, we consider an operator process $X=(X(t))_{t\geq0}$, defined on a complex separable Hilbert space $H$ containing a total invariant subset $D$ , with evolution described by a quantum stochastic differential equation of the form

\begin{eqnarray}
dX(t)&=&d\tau(t)(FX+Gu+L)(t)+\sum_{a\in I}dM_a(t)F_a(t)(wX+z)(t)\\
X(0)&=&C,\, 0\leq t \leq T < +\infty
\end{eqnarray}

or of the form

\begin{eqnarray}
dX(t)&=&-\{d\tau(t)(FX+Gu+L)(t)+\sum_{a\in I}dM_a(t)F_a(t)(wX+z)(t)\}\\
X(T)&=&C,\, 0\leq t \leq T < +\infty
\end{eqnarray}

where $I$ is a finite set, $M=\{M_a \,/\,a\in I\}$ is a self-adjoint  family of  regular integrators of scalar type  satisfying a $\rho$-commutation relation with It\^{o} multiplication rules

\begin{eqnarray}
&dM_a(t)\,dM_b(t)=\sum_{l \in I}c_{ab}^l (t)\,dM_l (t)&
\end{eqnarray}

and

\begin{eqnarray}
&dM_a (t)\,d\tau (t)=d\tau (t)\,dM_a (t)=0&
\end{eqnarray}

\medskip

where the $c_{ab}^l$ 's are the structure processes, $\tau $ is a real-valued measure on $[0,+\infty )$, absolutely continuous with respect to Lebesgue measure,    $dM_a \ne d\tau $ for all $a \in I$, and  the coefficient processes $F, G, u, L, w, z$ and $F_a$ for all $a \in I$, are as in Section 2. We assume also that $C$ is a bounded operator on $H$.  As discussed in section 5,  (7.54)- (7.55) and (7.56)-(7.57) admit unique locally bounded solutions. 

Under extra assumptions on the coefficients, e.g if $F,G,L,w,z$ are real-valued functions and the $F_{\alpha}$'s are complex-valued functions with conjugate $\overline{F_{\alpha}}=F_{\alpha^*}$ where $a^*$ is defined by $(M_a)^*=M_{a^*}$, then  the $X(t)\,'s$ correspond to classical quantum mechanical observables.

We associate with (7.54)- (7.55) and (7.56)-(7.57)  respectively the quadratic performance criteria (7.60) and (7.61)  of the following

\begin{definition}

For $\xi \in D$, $0\leq T < +\infty$, $Q,\,R,\,m,\,\eta$  locally bounded, strongly continuous, adapted processes such that $R$ has an inverse $R^{-1}$ with the same properties, for $Q_T,\,m_T,\,Q_0,\,m_0$   bounded operators on $H$ with $R(t)\geq 0$, $Q(t)\geq 0$, $R^{-1}(t)>0$ for all $t\in [0,T]$, and for $Q_T\geq 0$, $Q_0\geq 0$  define

\begin{eqnarray}
&\tilde J_{\xi,T}(u)=\int_0^T\,d\tau(t)(\langle X(t)\xi,Q(t)X(t)\xi\rangle + \langle u(t)\xi,R(t)u(t)\xi\rangle& \\
&+2\langle m(t)X(t)\xi,\xi\rangle+2\langle \eta(t)u(t)\xi,\xi\rangle)+\langle Q_T\,X(T)\xi,X(T)\xi\rangle +2\langle m_T\,X(T)\xi,\xi\rangle\nonumber&
\end{eqnarray}

and
 
\begin{eqnarray}
&\tilde J_{\xi,0}(u)=\int_0^T\,d\tau(t)(\langle X(t)\xi,Q(t)X(t)\xi\rangle + \langle u(t)\xi,R(t)u(t)\xi\rangle &\\
&+2\langle m(t)X(t)\xi,\xi\rangle+2\langle \eta(t)u(t)\xi,\xi\rangle)+\langle Q_0\,X(0)\xi,X(0)\xi\rangle +2\langle m_0\,X(0)\xi,\xi\rangle\nonumber&
\end{eqnarray}

\end{definition}

We view $u$ as a control process and we consider the problem of choosing it so as to minimize $\tilde J_{\xi,T}(u)$ (resp. $\tilde J_{\xi,0}(u)$), thus controlling the evolution of the process $X$.

\begin{theorem}
Let $T>0$ be a finite time, let $X=\{X(t)\,/\,t \geq  0\}$ be a locally bounded adapted process with evolution described by (7.54)- (7.55)  (resp.  (7.56)-(7.57) ) and with performance criterion   (7.60)   (resp.(7.61)), and suppose that there exists a self-adjoint, locally bounded process $ \Pi = \{ \Pi (t)\,/\,t \geq 0\}$  satisfying, weakly on the invariant domain $D$, the generalized  stochastic Riccati differential equation

\begin{eqnarray}
&d\tau(t)(F^*\Pi +\Pi F+Q-\Pi GR^{-1}G^*\Pi )(t)+[(\sum_{a\in I}dM_a\,F_aw)^* \,\Pi + &\\
&\Pi \,\sum_{a\in I}dM_a \,F_a w \pm (\sum_{a\in I}dM_a \,F_a w)^* \,\Pi \,(\sum_{a\in I}dM_a \,F_a w)](t) \pm & \nonumber\\
&(\sum_{a\in I}dM_a  \,F_aw \pm id)^*(t)\,d\Pi (t)\,(\sum_{a\in I} dM_a \,F_a w \pm id)(t)=0&\nonumber\\
&\Pi(T)=Q_T(\mbox{resp. }\,\Pi(0)=Q_0), 0\leq t \leq T &
\end{eqnarray}

and a locally bounded adapted process $r=\{r(t)\,/\,t \geq 0\}$ satisfying the stochastic differential equation

\begin{eqnarray}
&d\tau(t)(F^*r-\Pi GR^{-1}G^*r+\Pi L+m^*-\Pi GR^{-1}\eta^* )(t)+[(\sum_{a\in  I}dM_a\, F_a w)^* r+&\\
&\Pi \,\sum_{a\in I}dM_a \, F_a z + d\Pi \, \sum_{a\in I}dM_a \, F_a z \pm  (\sum_{a\in I}dM_a \, F_a w)^* \,\Pi \,\sum_{a\in I}dM_a \, F_a z \pm &\nonumber\\
&(\sum_{a\in I}dM_a\, F_a w)^*\,d\Pi \,\sum_{a\in I}dM_a \, F_a z ](t)+  [(\sum_{a\in I}dM_a \, F_a w \pm id)^* dr](t) =0&\nonumber\\
&r(T)=m_T^* (\mbox{resp.  }\,r(0)= m_0^* ) , 0\leq t \leq T,&
\end{eqnarray}

where $id$ denotes the identity operator on $H$, the plus sign in $\pm$ in (7.62) and (7.63) is associated with   (7.54)- (7.55)      and (7.60), and the minus with    (7.56)-(7.57) and (7.61). Then the performance criterion $\tilde J_{\xi,T}(u)$ (resp. $\tilde J_{\xi,0}(u)$) appearing in (7.60) (resp.  (7.61)) is minimized by the feedback control process

\begin{eqnarray}
u=-R^{-1}(G^*(\Pi X+r)+\eta^*).
\end{eqnarray}

Note: For $w=0$ and $z=id$ we obtain the solution to the quantum analogue of the "linear regulator" problem of classical control theory.

\end{theorem}

\begin{proof} We will give the proof for    (7.54)- (7.55)        and (7.60). The proof for  (7.56)- (7.57)        and (7.61) is similar. So let $u=\Lambda X+\lambda+ \mu$ where $\Lambda,\,\lambda$ are fixed processes to be chosen later and $\mu$ is the new control. We will choose $\Lambda,\,\lambda$ so that the minimizing new control $\mu$ is identically $0$. Replacing $u$ by $\Lambda X+\lambda+ \mu$ in (7.54) we obtain

\begin{eqnarray}
&dX(t)=d\tau (t)(FX+G\Lambda X+G\lambda +G\mu +L)(t)+\sum_{a\in I}dM_a(t)F_a(t)(wX+z)(t)&\\
&X(0)=C,\, 0\leq t \leq T < +\infty &
\end{eqnarray}

Let $Y$ be the solution of the above equation corresponding to $\mu=0$, i.e

\begin{eqnarray}
&dY(t)=d\tau (t)(FY+G\Lambda Y+G\lambda  +L)(t)+\sum_{a\in I}dM_a(t)F_a(t)(wY+z)(t)&\\
&Y(0)=C,\, 0\leq t \leq T < +\infty &
\end{eqnarray}

with corresponding control $u_0=\Lambda Y+\lambda$. Letting $\hat X=X-Y$ we obtain

\begin{eqnarray}
&d\hat X(t)=d\tau (t)(F \hat X+G\Lambda \hat X +G\mu)(t)+\sum_{a\in I}dM_a(t)F_a(t)(w \hat X)(t)&\\
&\hat X(0)=0,\, 0\leq t \leq T < +\infty &
\end{eqnarray}

and using $u=\Lambda \hat X+ u_0 +\mu$ (7.60) becomes

\begin{eqnarray}
&\tilde J_{\xi,T}(u)=\tilde J_{\xi,T}(u_0)+ \int_0^T\,d\tau (t)(\langle \xi,[\hat X^*Q \hat X +(\Lambda \hat X+\mu)^*R(\Lambda \hat X + \mu)](t)\xi \rangle& \\
&+\langle \xi,\hat X^*(T)Q_T\hat X(T)\xi \rangle +2\Re\,K&\nonumber
\end{eqnarray}

where $\Re\,K$ denotes the real part of $K$ and

\begin{eqnarray}
&K=\int_0^T\,d\tau (t)\langle \xi,[\hat X^*QY +(\Lambda \hat X+\mu)^*R(\Lambda  Y + \lambda) +\hat X^*m^*+(\Lambda \hat X+\mu)^*\eta^*](t)\xi\rangle& \\
& + \langle \xi,[\hat X^*(T)m_T^*+
\hat X^*(T)Q_TY(T)]\xi \rangle&\nonumber
\end{eqnarray}

We will show that if 

\begin{eqnarray}
&\Lambda=-R^{-1}G^*\Pi,\, \lambda =-R^{-1}(G^*r+\eta^*)&
\end{eqnarray}

 then $K=0$. In view of (7.73) we will then conclude that (7.60) is minimized by $\mu=0$. In so doing, let $p(t)=r(t)+\Pi(t)Y(t)$. Then

\begin{eqnarray}
& \langle \xi,[\hat X^*(T)m_T^*+\hat X^*(T)Q_TY(T)]\xi \rangle =\int_0^T\,d\langle \xi,(\hat X^*p)(t)\xi \rangle &
\end{eqnarray}

Using

\begin{eqnarray}
&d\hat X^*(t)=d\tau (t)( \hat X ^* F^* + \hat X ^*{\Lambda}^* G^* +{\mu}^* G^*)(t)+\sum_{a\in I}dM_{a^*}(t){\rho}_{a^*} ( \hat X^* w^* F_a^*)(t)&\\
&\hat X^*(0)=0,\, 0\leq t \leq T < +\infty &
\end{eqnarray}

where $a^*$ is defined by $(M_a)^*=M_{a^*}$, and (5.21) to compute the right hand side of (7.76),  (7.74) becomes

\begin{eqnarray}
&K=\int_0^T\,<\xi, \{\hat X^* [d\tau ((F^* \Pi +\Lambda^* G^* \Pi +\Pi (F+G \Lambda)+Q+\Lambda^* R \Lambda)+&\\
&\sum_{a \in I} w^* F_a^* \,dM_{a^*} \, \Pi + \Pi  \, \sum_{a \in I}dM_a \,F_a w +d\Pi \, \sum_{a \in I}dM_a \,F_a w+&\nonumber\\
 &\sum_{a \in I}w^* F_a^*\,dM_{a^*}\,d\Pi+ \sum_{a \in I}w^* F_a^*\,dM_{a^*}\,\Pi\,\sum_{a \in I}dM_a \,F_a w+&\nonumber\\
& \sum_{a \in I}w^* F_a^*\,dM_{a^*}\,d\Pi \,\sum_{a \in I}dM_a \,F_a w+d\Pi)Y+(d\tau(F^* r+\Lambda^* G^* r+&\nonumber\\
&\Pi(G \lambda+L)+\Lambda^* R \lambda+m^*+\Lambda^* \eta^* )+\sum_{a \in I}w^* F_a^* \,
dM_{a^*}\,r+\Pi \sum_{a \in I}dM_a \,F_a z+&\nonumber\\
&d\Pi \,\sum_{a \in I}dM_a \,F_a z +\sum_{a \in I}w^* F_a^* \,dM_{a^*}\,dr+\sum_{a \in I}w^* F_a^*\,
dM_{a^*}\,\Pi\,\sum_{a \in I}dM_a \,F_a z +&\nonumber\\
&\sum_{a \in I}w^* F_a^*\,
dM_{a^*}\,d\Pi \,\sum_{a \in I}dM_a \,F_a z+dr)]+\mu^* [(G^* \Pi+R^* \Lambda )Y+&\nonumber\\
&(G^*r+R \lambda+\eta^*)]\,d\tau \}(t)\xi >&\nonumber
\end{eqnarray}

Replacing in the above $\Lambda$ and $\lambda$ by $-R^{-1}G^*\Pi$ and $-R^{-1}(G^*r+\eta^*)$ respectively we see that the coefficients of $Y$ are, in view of (7.62)-(7.63) , equal to zero. The same is true, by (7.64)-(7.65), for the constant terms. Thus $K=0$.

\end{proof}
 
\medskip

\begin{definition} The family of integrators of scalar type  $\{M_0,M_a\,/\,a\in I\}$ where $dM_0=d\tau$, is said to be linearly independent if the equality

\begin{eqnarray}
 \int_0^t\,d\tau(s)\,G(s)+\sum_{a\in I}\,dM_a(s)\,G_a(s)=0
\end{eqnarray}

 for all families $\{G,G_a\,/\,a\in I\}$ of adapted processes and all $t\geq 0$, implies that $G=G_a=0$ for all $a\in I$.
\end{definition}

\begin{proposition}
 If the family of integrators of scalar type appearing in     (7.54)- (7.55)          and    (7.56)- (7.57)          is linearly independent then the Riccati equation  (7.62)- (7.63)       can be put in the form 

\begin{eqnarray}
d\Pi (t)= d\tau (t)\,A(t,\Pi (t))+ \sum_{a\in I}\,dM_a(t)\,B_a(t,\Pi (t))
\end{eqnarray}

 where $A(t)=A(t,\Pi (t))$ and $B_a(t)=B_a(t,\Pi (t))$ can be described as the solutions of the operator equations

\begin{eqnarray}
&(F^*\Pi +\Pi F+Q-\Pi GR^{-1}G^*\Pi \pm A \pm \sum_{a,b\in I}[c_0(a,b)\rho _b(\rho _a(w^* F_{j(a)}^*)\Pi )F_b w\pm &\\
 &c_0(a,b)\rho _b(\rho _a(w^*F_{j(a)}^*))B_b\pm c_0(a,b)\rho _b (B_a)F_b w+\sum_{\gamma,\epsilon \in I}c_{\epsilon}(a,b)c_0(\epsilon,\gamma)&\nonumber\\
&\rho _{\gamma}(\rho _b(\rho _a(w^*F_{j(a)}^*))B_b)F_{\gamma} w] )(t)=0,\,\,\, \forall t \in [0,T]&\nonumber  
\end{eqnarray}

and for all $J \in I$  

\begin{eqnarray}
&(\rho _J(w^*F_{j(J)}^*)\Pi +\rho_J (\Pi) F_J w \pm B_J + \sum_{a,b\in I}[c_J (a,b)\rho _b (B_a) F_b w +&\\
& c_J (a,b) \rho_b (\rho_a(w^*F_{j(a)}^*))B_b \pm c_J(a,b)\rho_b (\rho_a (w^* F_{j(a)}^*)\Pi )F_b W \pm \sum_{\gamma,\epsilon \in I} &\nonumber\\
&\rho_J(c_{\epsilon}(a,b))c_J(a,\gamma)\rho _{\gamma}(\rho _b(\rho _a(w^*F_{j(a)}^*))B_b)F_{\gamma} w])(t)=0,\,\,\, \forall t \in [0,T]&\nonumber
\end{eqnarray}

while  (7.64)- (7.65)   can be put in the form

\begin{eqnarray}
dr(t)=d\tau (t)\,C(t,r(t))+\sum_{a\in I}\,dM_a(t)\,D_a(t,r(t))
\end{eqnarray}

 where $C(t)=C(t,r(t))$ and $D_a(t)=D_a(t,r(t))$ can be described as the solutions of the operator equations

\begin{eqnarray}
&(F^*r -\Pi GR^{-1}G^* r+\Pi L+m^* -\Pi GR^{-1} {\eta}^*  \pm C+ \sum_{a,b\in I}[c_0(a,b)\rho _b (B_a)F_b z & \\
 &\pm c_0(a,b)\rho _b(\rho _a(w^* F_{j(a)}^*)\Pi )F_b z+
 c_0(a,b)\rho _b(\rho _a(w^*F_{j(a)}^*))D_b& \nonumber\\
 &\pm \sum_{\gamma,\epsilon \in I}c_{\epsilon}(a,b)c_0(\epsilon,\gamma)])(t)=0, \forall t \in [0,T]&\nonumber  
\end{eqnarray}

and for all $J \in I$  

\begin{eqnarray}
&(\rho _J(w^*F_{j(J)}^*)r +\rho_J (\Pi) F_J z \pm D_J + \sum_{a,b\in I}[c_J (a,b)\rho _b (B_a) F_b z +& \\
 &\rho_b (\rho_a(w^*F_{j(a)}^*)\Pi)F_b z +\rho_b (\rho_a (w^* F_{j(a)}^*) )D_b ) \pm \sum_{\gamma,\epsilon \in I} \rho_J(c_{\epsilon}(a,b))c_J(\epsilon,\gamma)&\nonumber\\
&\rho _{\gamma}(\rho _b(\rho _a(w^*F_{j(a)}^*)))\rho _{\gamma} (B_b)F_{\gamma}z])(t)=0, \forall t \in [0,T]&\nonumber
\end{eqnarray}

Here the adapted processes $c_0(a,b),c_{\epsilon}(a,b),a,b\in I$ are defined for all $t\geq 0$ by

\begin{eqnarray}
dM_a(t)\,dM_b(t)=d\tau (t) \,c_0(a,b)(t)+ \sum_{\epsilon \in I}\,dM_{\epsilon}(t)\,c_{\epsilon}(a,b)(t)
\end{eqnarray}

 $\rho_a$ is, for each $a\in I$, the corresponding commutation homomorphism, and the mapping $j:I\rightarrow I$ is defined by $j(a^*)=a$.

\end{proposition}

\begin{proof} We will only give the proof for $ \Pi = \{ \Pi (t)\,/\,t \geq 0\}$ . The proof for   $r=\{r(t)\,/\,t \geq 0\}$ is similar. Substituting  $d\Pi = d\tau \,A+ \sum_{a \in I}\,dM_a \, B_a$ in (7.62) we obtain

\begin{eqnarray}
&\{d\tau\,(F^*\Pi +\Pi F+Q-\Pi GR^{-1}G^*\Pi)+\sum_{a \in I}\,dM_{a^*} \,\rho_{a^*}\,(w^*F^*_a)\Pi+&\\
&\Pi \sum_{a \in I}\,dM_a \, F_a w \pm \sum_{a \in I}\,dM_{a^*} \,\rho_{a^*} \,(w^*F^*_a)\Pi \sum_{b \in I}\,dM_b \, F_b w + &\nonumber\\
&\sum_{a \in I}\,dM_{a^*} \,\rho_{a^*} \,(w^*F^*_a)(d \tau \,A + \sum_{b \in I}\,dM_b \,B_b)+&\nonumber\\
&(d \tau \,A+\sum_{a \in I}\,dM_a \,B_a)\sum_{b \in I}\,dM_b \,F_b w &\nonumber\\
&\pm \sum_{a \in I}\,dM_{a^*} \,\rho_{a^*}\,(w^*F^*_a)(d\tau \,A+ \sum_{b \in I}\,dM_b \, B_b)\sum_{\gamma \in I}\,dM_{\gamma} \, F_{\gamma} w &\nonumber
\\
&\pm d\tau \,A \pm \sum_{a \in I}\,dM_a \, B_a \}(t)=0&\nonumber
\end{eqnarray}

Making use of (7.59), of the $\rho$-commutation relations,  of (7.87), and of

\begin{eqnarray}
dM_a(t)\,dM_b(t)\,dM_{\gamma}(t)=\sum_{\epsilon, \delta \in I}\,dM_{\epsilon}(t)\,c_{\epsilon}(\delta,\gamma)(t)\rho_{\gamma}(c_{\delta}(a,b)(t))
\end{eqnarray}

we obtain after renaiming the indices,

\begin{eqnarray}
&\{d\tau\,(F^*\Pi +\Pi F+Q-\Pi GR^{-1}G^*\Pi \pm A \pm \sum_{a,b \in I} [c_0(a,b)\rho_b(\rho_a(w^*F^*_{j(a)})\Pi)F_bw &\\
&\pm c_0(a,b)\rho_b(\rho_a(w^*F^*_{j(a)}))B_b \pm c_0(a,b)\rho_b(B_a)F_bw +&\nonumber
\\
&\sum_{\gamma,\epsilon}c_{\epsilon}(a,b)c_0(\epsilon,\gamma) \rho_{\gamma} (\rho_b (\rho_a(w^*F^*_{j(a)}))B_b)F_{\gamma}w])+&\nonumber
\\
& \sum_{J \in I}\,dM_J(\rho_J(w^* F^*_{j(J)})\Pi+\rho_J(\Pi)F_Jw \pm B_J +\sum_{a,b \in I} [c_J(a,b) \rho_b(B_a)F_bw +&\nonumber
\\
&c_J(a,b) \rho_b(\rho_a(w^*F^*_{j(a)}))B_b \pm c_J(a,b)\rho_b(\rho_a(w^*F^*_{j(a)})\Pi)F_bw \pm  &\nonumber
\\
&\sum_{\epsilon, \gamma \in I} \rho_J ( c_{\epsilon}(a,b))c_J(a,\gamma)\rho_{\gamma} (\rho_{b} (\rho_a(w^*F^*_{j(a)}))B_b)F_{\gamma}w])\}(t)=0&\nonumber
\end{eqnarray}

from which the result follows by the linear independence assumption.

\end{proof}

\section{Quantum stochastic Riccati equations}

\begin{definition} In the notation of the representation free calculus described in  section 5, let $M,M_1,M_2$ be integrator processes of scalar type such that $dM_1=dM$ and $dM_2=dM^*$. If the Meyer bracket $[[M_b,M_a]]$ exists for all  $a,b\in \{1,2\}$ and is a complex-valued nonatomic measure, then the pair $(M_1,M_2)$ is called a Levy pair.

 A Levy pair $(M_1,M_2)$ is said to be of Boson type if $[M_b(I),M_a(J)]=0$ for all  $a,b\in \{1,2\}$ and $I,J \subset [0,+\infty )$ with $I \cap J= \varnothing$, and it is said to be of Fermion type if $\{M_b(I),M_a(J)\}=0$ where $[x,y]=xy-yx$ and $\{x,y\}=xy+yx$. 

For a Boson (resp. Fermion) type Levy pair, $\rho_1=\rho_2=id$ 
(resp. $\rho_1=\rho_2=-id$) where $\rho_1,\rho_2$ are the commutation automorphisms corresponding to $M_1,M_2$ respectively. 

For a Levy pair $(M_1,M_2)$, $dM_b^*(t)\,dM_a(t)=\sigma_{ba}(t)\,dt$, for all  
$a,b\in \{1,2\}$, where the matrix valued function $t\rightarrow (\sigma_{ba}(t))_{a,b\in \{1,2\}}$ is positive definite in the sense that for all complex-valued  continuous functions $f$, $(f\,\,f)\cdot \sigma \cdot (f\,\,f)^t \geq 0$. 

Let $M_0$ be an integrator of scalar type such that $dM_0 (t)=dt$ where $dt$ is the usual time differential. If $M_0,M_1,M_2$ are linearly independent in the sense of Definition 3, then $(M_1,M_2)$ is called a linearly independent Levy pair.

\end{definition}

If $\{M_a\,/\,a \in I\}=\{M_1,M_2\}$ is a linearly independent Levy pair (e.g a Boson or a Fermion Levy pair) then Theorem 10 includes the solution to the control problem of stochastic evolutions driven by quantum Brownian motion,  in terms of the solution of a stochastic Riccati equation to be studied in more detail in this section.  The Riccati equation (7.62)-(7.63) reduces to

\begin{eqnarray}
&d\Pi (t)=dt((\mp F+\sigma _{11} \rho _2 (F_2 w) \rho _1 \rho _2 (F_1 w)+\sigma _{12} 
\rho _2 (F_2 w)F_2 w+&\\
&\sigma _{22} 
\rho _1 (F_1 w)\rho _2 \rho _1 (F_2 w)+\sigma _{21 }\rho _1 (F_1 w) F_1 w)^* \Pi + (\mp \Pi F +\sigma_{11} 
 \rho _1 \rho _2 (\Pi )\rho _1 (F_2 w)F_1 w\nonumber\\
&+\sigma _{12} \Pi \rho _2 (F_2 w)F_2 w +\sigma _{22} \rho _2 \rho _1 (\Pi )\rho _2(F_1 w)F_2 w+\sigma _{21} \Pi  \rho _1 (F_1 w)F_1 w)+& \nonumber\\
&\sigma _{11}\rho _1 \rho _2( w^* F_1^* )\rho _1(\Pi )F_1 w
+\sigma _{12} w^* F_1^* \rho _2 (\Pi )  F_2 w  + \sigma _{22} 
\rho _2 \rho _1 (w^* F_2^* )\rho _2(\Pi )F_2 w& \nonumber\\
&+\sigma_{21} w^* F_2^* \rho_1  ( \Pi) F_1 w \mp Q \pm \Pi GR^{-1}G^* \Pi )(t) \mp  dM_1(t)(\rho _1(w^* F_2^* )\Pi  & \nonumber\\
& +\rho _1 (\Pi) F_1 w)(t) \mp  dM_2(t)(\rho _2(w^* F_1^* )\Pi+  \rho _2 (\Pi ) F_2 w)(t)& \nonumber\\
&\Pi (T)=Q_T\,\,(\mbox{resp.}\,\Pi (0)=Q_0), 0\leq t \leq T& 
\end{eqnarray}   

where the plus (resp. minus) sign in $\pm$ (resp. in $\mp$) is associated with (7.54)-(7.55) and (7.60), and the minus (resp. plus) with (7.56)-(7.57)  and (7.61).

\begin{theorem}
The Riccati equation  (8.1)-(8.2)  admits a unique, adapted, strongly continuous, positive, locally bounded, solution $\Pi =(\Pi (t))_{0 \leq t \leq T}$ defined weakly on the invariant domain $D$.
\end{theorem}

\begin{proof}
We first consider the case corresponding to (7.54)-(7.55) and (7.60). Equation (8.1)-(8.2) can be written as

\begin{eqnarray}
&d\Pi (t)=[dt( F+\sigma _{11} \rho _2 (F_2 w) \rho _1 \rho _2 (F_1 w)+\sigma _{12} 
\rho _2 (F_2 w)F_2 w+&\\ 
 & \sigma _{22} \rho _1 (F_1 w)\rho _2 \rho _1 (F_2 w)
+\sigma _{21 }\rho _1 (F_1 w) F_1 w- GR^{-1}G^* \Pi )(t)+dM_1(t)( F_1 w)(t)&\nonumber\\
& + dM_2(t)( F_2 w )(t)]^*\Pi (t)+[dt ( \Pi F +\sigma_{11} \rho _1 \rho _2 (\Pi )\rho _1 (F_2 w)F_1 w+&\nonumber\\
&\sigma _{12} \Pi \rho _2 (F_2 w)F_2 w + \sigma _{22} \rho _2 \rho _1 (\Pi )\rho _2(F_1 w)F_2 w +\sigma _{21} \Pi  \rho _1 (F_1 w)F_1 w-&\nonumber\\
&\Pi GR^{-1}G^* \Pi )(t)+dM_1(t)(\rho _1(\Pi ) F_1 w)(t)
 + dM_2(t)(\rho _2(\Pi ) F_2 w )(t)]&\nonumber\\
&+dt[\sigma _{11}\rho _1 \rho _2( w^* F_1^* )\rho _1(\Pi )F_1 w 
+\sigma _{12} w^* F_1^* \rho _2 (\Pi )  F_2 w &\nonumber\\
 & + \sigma _{22} \rho _2 \rho _1 (w^* F_2^* )\rho _2(\Pi )F_2 w+\sigma_{21} w^* F_2^* \rho_1  ( \Pi) F_1 w ](t) +dt( Q + \Pi GR^{-1}G^* \Pi )(t)&\nonumber\\
&\Pi (0)=Q_0\,\,, 0\leq t \leq T &
\end{eqnarray}

Using (5.21), the identities  $dt\,dM_1=dt\,dM_2=dM_1\,dt=dM_2\,dt=0$, $dM_1\,dM_2=dM_2^*\,dM_2=\sigma _{22} dt$, $ dM_2\,dM_1=dM_1^*\,dM_1=\sigma _{11} dt$,  $dM_2\,dM_2=dM_1^*\,dM_2=\sigma _{12} dt$, $dM_1\,dM_1=dM_2^*\,dM_1=\sigma _{21} dt $, and the fact that if $\{ \lambda (t,s)\,/\,t \mbox{(resp.}s) \geq 0 \}$ is for each $s$ (resp. $t$) a process then 

\begin{eqnarray}
d\left(\int_0^t\,ds\,\lambda (t,s)\right)=dt\,\lambda (t,t)+\int_0^t\,ds\, \,\,d\lambda (t,s)   
\end{eqnarray}

we can prove by taking the time differential of the right hand side and showing that it satisfies (8.3)-(8.4), that weakly on $D$

\begin{eqnarray}
\Pi (t)=K(t,0)Q_0K(t,0)^*+\int_0^t\,ds\,K(t,s)(Q + \Pi GR^{-1}G^* \Pi )(s)K(t,s)^*
\end{eqnarray}

where

\begin{eqnarray}
&dK(t,s)=[dt( F+\sigma _{11} \rho _2 (F_2 w) \rho _1 \rho _2 (F_1 w)+\sigma _{12} 
\rho _2 (F_2 w)F_2 w+&\\
&\sigma _{22} \rho _1 (F_1 w)\rho _2 \rho _1 (F_2 w)+\sigma _{21 }\rho _1 (F_1 w) F_1 w- GR^{-1}G^* \Pi )(t)&\nonumber\\
&+dM_1(t)( F_1 w)(t)\ + dM_2(t)( F_2 w )(t)]^* K(t,s)&\nonumber\\
&K(s,s)=id,\,\,s\leq t \leq T &
\end{eqnarray}

Let the sequence $\{ \Pi _n \} _{n=1}^{+\infty}$ of locally bounded self-adjoint processes, be defined by the iteration scheme

\begin{eqnarray}
&\Pi _1 (t) = Q_0&
\end{eqnarray}

and for $n \geq 1$

\begin{eqnarray}
&\Pi _{n+1} (t)=K_n(t,0)\,Q_0\,K_n(t,0)^*+\int_0^t\,ds\,K_n(t,s)(Q + \Pi _n GR^{-1}G^* \Pi _n )(s)K_n(t,s)^*&
\end{eqnarray}

where $K_n(t,s)$ is the unique locally bounded solution of

\begin{eqnarray}
&dK_n(t,s)=[dt( F+\sigma _{11} \rho _2 (F_2 w) \rho _1 \rho _2 (F_1 w)+\sigma _{12} 
\rho _2 (F_2 w)F_2 w&\\
&+\sigma _{22} 
\rho _1 (F_1 w)\rho _2 \rho _1 (F_2 w)
+\sigma _{21 }\rho _1 (F_1 w) F_1 w- GR^{-1}G^* \Pi )(t)&\nonumber\\
&+dM_1(t)( F_1 w)(t)\ + dM_2(t)( F_2 w )(t)]^* K_n(t,s)&\nonumber\\
&K_n(s,s)=id,\,\,s\leq t \leq T & 
\end{eqnarray}

Since $Q_0\geq 0$, $Q(t)\geq 0$, and $R^{-1}(t) >0$, it follows from the above equation that

\begin{eqnarray}
&\Pi _n (t) \geq 0,\mbox{  for all }n=1,2,...,\mbox{ and }t\in [0,T]&
\end{eqnarray}

Moreover, for all $t\in [0,T]$ and $n=2,3,...$

\begin{eqnarray}
&0 \leq \Pi _{n+1} (t) \leq \Pi _n (t)&
\end{eqnarray}

To prove this, we notice that

\begin{eqnarray}
&d\Pi _{n+1} (t)=(dt( F+\sigma _{11} \rho _2 (F_2 w) \rho _1 \rho _2 (F_1 w)+\sigma _{12}\rho _2 (F_2 w)F_2 w+&\\
&\sigma _{22}\rho _1 (F_1 w)\rho _2 \rho _1 (F_2 w)
+\sigma _{21 }\rho _1 (F_1 w) F_1 w- GR^{-1}G^* \Pi _n )(t)&\nonumber\\
&+dM_1(t)( F_1 w)(t) + dM_2(t)( F_2 w )(t))^*\Pi _{n+1} (t)+&\nonumber\\
&(dt ( \Pi _{n+1} F +\sigma_{11} \rho _1 \rho _2 (\Pi _{n+1} )\rho _1 (F_2 w)F_1 w&\nonumber\\
&+\sigma _{12} \Pi _{n+1} \rho _2 (F_2 w)F_2 w
 + \sigma _{22} \rho _2 \rho _1 (\Pi _{n+1} )\rho _2(F_1 w)F_2 w &\nonumber\\
&+\sigma _{21} \Pi _{n+1}  \rho _1 (F_1 w)F_1 w-\Pi _{n+1} GR^{-1}G^* \Pi _n )(t)&\nonumber\\
&+dM_1(t)(\rho _1(\Pi _{n+1} ) F_1 w)(t) + dM_2(t)(\rho _2(\Pi _{n+1} ) F_2 w )(t))&\nonumber\\
& +dt(Q+\Pi _n  GR^{-1}G^* \Pi _n)(t)
+dt(\sigma _{11}\rho _1 \rho _2( w^* F_1^* )\rho _1(\Pi  _{n+1})F_1 w &\nonumber\\
&+\sigma _{12} w^* F_1^* \rho _2 (\Pi _{n+1} )  F_2 w 
 + \sigma _{22}\rho _2 \rho _1 (w^* F_2^* )\rho _2(\Pi _{n+1} )F_2 w&\nonumber\\
&+\sigma_{21} w^* F_2^* \rho_1  ( \Pi _{n+1}  ) F_1 w )(t) &\nonumber\\
&\Pi _{n+1} (0)=Q_0\,\,, 0\leq t \leq T &
\end{eqnarray}

Letting $P_n(t)=\Pi _{n+1} (t) - \Pi _{n} (t)$ the above yields

\begin{eqnarray}
&dP _{n} (t)=(dt( F+\sigma _{11} \rho _2 (F_2 w) \rho _1 \rho _2 (F_1 w)+\sigma _{12}\rho _2 (F_2 w)F_2 w+&\\
&\sigma _{22}\rho _1 (F_1 w)\rho _2 \rho _1 (F_2 w)
+\sigma _{21 }\rho _1 (F_1 w) F_1 w- GR^{-1}G^* \Pi _n )(t)+&\nonumber\\
&dM_1(t)( F_1 w)(t) + dM_2(t)( F_2 w )(t))^* P _{n} (t)&\nonumber\\
&+(dt ( P _{n} F +\sigma_{11} \rho _1 \rho _2 (P _{n} )\rho _1 (F_2 w)F_1 w+&\nonumber\\
&\sigma _{12} P _{n} \rho _2 (F_2 w)F_2 w
 + \sigma _{22} \rho _2 \rho _1 (P _{n} )\rho _2(F_1 w)F_2 w+&\nonumber\\
&\sigma _{21} P _{n}  \rho _1 (F_1 w)F_1 w-P _{n} GR^{-1}G^* \Pi _n )(t)&\nonumber\\
&+dM_1(t)(\rho _1(P_{n} ) F_1 w)(t) + dM_2(t)(\rho _2(P _{n} ) F_2 w )(t))&\nonumber\\
&+dt(\sigma _{11}\rho _1 \rho _2( w^* F_1^* )\rho _1(P  _{n})F_1 w&\nonumber\\ 
&+\sigma _{12} w^* F_1^* \rho _2 (P _{n} )  F_2 w  + \sigma _{22}\rho _2 \rho _1 (w^* F_2^* )\rho _2(P _{n} )F_2 w&\nonumber\\
&+\sigma_{21} w^* F_2^* \rho_1  ( P _{n}  ) F_1 w )(t) -dt(P _{n-1}  GR^{-1}G^* P _{n-1})(t) &\nonumber\\
&P _{n} (0)=0\,\,, 0 \leq t \leq T &
\end{eqnarray}

Thus as we did for $\Pi (t)$, for all $n=2,3,...$ and  $t \in [0,T]$ 

\begin{eqnarray}
&P _{n} (t)=-\int_0^t\,ds\,K_n(t,s)( P _{n-1} GR^{-1}G^* P _{n-1} )(s)K_n(t,s)^*&
\end{eqnarray}

weakly on $D$. Since $R^{-1}(s) >0$ for all $s\in [0,T]$, this implies that 

\begin{eqnarray}
P _{n} (t) \leq 0
\end{eqnarray}

 thus proving (8.14). By (8.14) $\{ \Pi _n (t) \} _{n=1}^{+\infty}$ converges strongly on the invariant domain $D$ and the convergence is uniform on compact $t$-intervals. Let $\Pi=\{ \Pi  (t) \} _{0 \leq t \leq T}$ denote the limit process. Being a strong limit of a decreasing sequence of adapted, strongly continuous, positive processes, $\Pi$ has the same properties. By the uniformity of the convergence of the defining sequence  $\{ \Pi _n (t) \} _{n=1}^{+\infty}$ and the arbitrariness of $T$, $\Pi$ is locally bounded. As above we can show that for $n=1,2,...$ 

\begin{eqnarray}
&\Pi _{n+1} (t)=\Phi (t,0)\,Q_0\,\Phi (t,0)^*+\int_0^t ds\,\Phi(t,s)(Q -\Pi _n GR^{-1}G^* \Pi _n + P_n GR^{-1}G^* P_n)(s)\Phi (t,s)^*&
\end{eqnarray}

weakly on the invariant domain $D$, where $\Phi(t,s)$ is the locally bounded solution of

\begin{eqnarray}
&d\Phi (t,s)=[dt( F+\sigma _{11} \rho _2 (F_2 w) \rho _1 \rho _2 (F_1 w)+\sigma _{12} 
\rho _2 (F_2 w)F_2 w+&\\
&\sigma _{22} 
\rho _1 (F_1 w)\rho _2 \rho _1 (F_2 w)
+\sigma _{21 }\rho _1 (F_1 w) F_1 w )(t)+dM_1(t)( F_1 w)(t)\ + dM_2(t)( F_2 w )(t)]^* \Phi (t,s)&\nonumber\\
&\Phi (s,s)=id,\,\,s\leq t \leq T. &
\end{eqnarray}

Let $h,\xi \in D$. By the uniformity of the convergence of $\Pi _n (t) \rightarrow \Pi (t)$ and \break $P _n (t) \rightarrow 0$, and by the local boundedness of $\Pi$, upon letting $n\rightarrow +\infty$ we obtain by the bounded convergence theorem and

\begin{eqnarray}
&&\\
&\langle \Pi _{n+1} (t) h, \xi \rangle=\langle [\Phi (t,0)\,Q_0\,\Phi (t,0)^*+\int_0^t ds\,\Phi(t,s)(Q -\Pi _n GR^{-1}G^* \Pi _n+ P_n GR^{-1}G^* P_n)(s)\Phi (t,s)^* ],h, \xi \rangle&\nonumber
\end{eqnarray}

that

\begin{eqnarray}
&\langle \Pi  (t) h, \xi \rangle=\langle [\Phi (t,0)\,Q_0\,\Phi (t,0)^*+\int_0^t ds\,\Phi(t,s)(Q -\Pi  GR^{-1}G^* \Pi)(s)\Phi (t,s)^* ],h, \xi \rangle&
\end{eqnarray}

from which, by the arbitrariness of $h, \xi$ we have that

\begin{eqnarray}
 &\Pi  (t) =\Phi (t,0)\,Q_0\,\Phi (t,0)^*+\int_0^t ds\,\Phi(t,s)(Q -\Pi GR^{-1}G^* \Pi)(s)\Phi (t,s)^* &
\end{eqnarray}

for all $t \in [0,T] $, weakly on $D$. By taking the differential of both sides of the above we can show that $\Pi$ solves (8.1)-(8.2). To see that such $\Pi$ is unique, let $\hat \Pi$ be another solution of (8.1)-(8.2). Letting $P(t)=\Pi(t)-\hat \Pi (t)$ we obtain

\begin{eqnarray}
&dP (t)=[dt( F+\sigma _{11} \rho _2 (F_2 w) \rho _1 \rho _2 (F_1 w)+\sigma _{12} 
\rho _2 (F_2 w)F_2 w+&\\
&\sigma _{22}\rho _1 (F_1 w)\rho _2 \rho _1 (F_2 w)+\sigma _{21 }\rho _1 (F_1 w) 
F_1 w- GR^{-1}G^* \hat  \Pi )(t)&\nonumber\\
&+dM_1(t)( F_1 w)(t) + dM_2(t)( F_2 w )(t)]^* P (t)+&\nonumber\\
&[dt ( P F +\sigma_{11} \rho _1 
\rho _2 (P )\rho _1 (F_2 w)F_1 w+&\nonumber\\
&\sigma _{12} P \rho _2 (F_2 w)F_2 w + \sigma _{22} \rho _2 \rho _1 (P )\rho _2(F_1 w)
F_2 w +&\nonumber\\
&\sigma _{21} P  \rho _1 (F_1 w)F_1 w-P GR^{-1}G^* \hat \Pi  )(t)&\nonumber\\
&+dM_1(t)(\rho _1(P ) F_1 w)(t) + dM_2(t)(\rho _2 (P) F_2 w )(t)]+&\nonumber\\
&dt[\sigma _{11}\rho _1 \rho _2( w^* F_1^* )\rho _1(P )F_1 w+\sigma _{12} w^* F_1^* \rho _2 (P )  F_2 w  +&\nonumber\\
& \sigma _{22}\rho _2 \rho _1 (w^* F_2^* )\rho _2(P )F_2 w+\sigma_{21} w^* F_2^* \rho_1  ( P) F_1 w ](t)&\nonumber\\
& -dt(P GR^{-1}G^* P )(t)&\nonumber\\
&P (0)=0\,\,, 0\leq t \leq T &
\end{eqnarray}

Thus, as before, $P (t) \leq 0$ i.e $\Pi (t) \leq \hat \Pi (t)$. By interchanging $\Pi$ and $\hat \Pi$ in  $P(t)=\Pi(t)-\hat \Pi (t)$ and replacing $\hat \Pi$ by $\Pi$ in the above equation we obtain that $\hat \Pi (t) \leq \Pi (t)$. Thus $\Pi(t)=\hat \Pi (t)$ which proves uniqueness.

 We now turn to the case of the Riccati equation corresponding to  (7.56)-(7.57)  and (7.61) which will be treated by using the, just proved, case corresponding to  
  (7.56)-(7.57)  and (7.61)   and reversing the time flow. So let $s=T-t$ in the Boson version of (8.1)-(8.2) and let, for an operator process $K$, $\hat K (s)=K(T-s)$ to obtain

\begin{eqnarray}
&d\hat \Pi (s)=ds((\hat F+\tilde \sigma _{11} \rho _2 (\hat F_2 \hat w) \rho _1 \rho _2 (\hat F_1 \hat w)+\tilde \sigma _{12} 
\rho _2 (\hat F_2 \hat w)\hat F_2 \hat w+&\\
&\tilde \sigma _{22} 
\rho _1 (\hat F_1 \hat w) \rho _2 \rho _1 (\hat F_2 \hat w)
+\tilde \sigma _{21 }\rho _1 (\hat F_1 \hat w) \hat F_1 \hat w)^* \hat \Pi + &\nonumber\\
&( \hat \Pi \hat F +\tilde \sigma_{11} 
 \rho _1 \rho _2 (\hat \Pi )\rho _1 (\hat F_2 \hat w)\hat F_1 \hat w+&\nonumber\\
&\tilde \sigma _{12} \hat \Pi \rho _2 (\hat F_2 \hat w)\hat F_2 \hat w
+\tilde \sigma _{22} \rho _2 \rho _1 (\hat \Pi )\rho _2(\hat F_1 \hat w)\hat F_2 \hat w+&\nonumber\\
&\tilde \sigma _{21} \hat \Pi  \rho _1 (\hat F_1 \hat w) \hat F_1 w))+\tilde \sigma _{11} \rho _1 \rho _2( {\hat w}^*  {\hat F}_1^* )\rho _1(\hat \Pi ) \hat F_1 \hat w &\nonumber\\
&+\tilde  \sigma _{12} {\hat w}^* {\hat F}_1^* \rho _2 (\hat \Pi )\hat  F_2
 \hat  w  +\tilde \sigma _{22} 
\rho _2 \rho _1 ({\hat w}^* {\hat F}_2^* )\rho _2(\hat \Pi )\hat F_2 \hat w&\nonumber\\
&+\tilde \sigma_{21}{\hat w}^* {\hat F}_2^* \rho_1  (\hat \Pi)\hat F_1 \hat w +\hat Q -\hat \Pi \hat G {\hat R}^{-1} {\hat G}^* \hat \Pi )(s) &\nonumber\\
&+  dN_1(s)(\rho _1 ({\hat w}^* {\hat F}_2^* )\hat \Pi +  \rho _1 (\hat \Pi) \hat F_1 \hat w)(s)&\nonumber \\
&+  dN_2(s)(\rho _2({\hat w}^* {\hat F}_1^* )\hat \Pi+  \rho _2 (\hat \Pi )\hat F_2 \hat w)(s)&\nonumber \\
&\hat \Pi (0)=Q_T\,\,, 0\leq s \leq T& 
\end{eqnarray}

where the Levy-pair $(N_1,N_2)$ is defined by 

\begin{eqnarray}
&N_1(s)=-M_1(T-s),\,\,\,N_2(s)=-M_2(T-s)&
\end{eqnarray}

with corresponding It\^{o} table

 \begin{eqnarray}
dN_b^* (s)\,dN_a(s)=\tilde {\sigma}_{ba} (s)\,ds
\end{eqnarray}

where $a,b \in \{1,2\}$ and
 
\begin{eqnarray}
&\tilde {\sigma}_{ba} (s)=-\sigma_{ba} (T-s)&.
\end{eqnarray}

Since the above differential equation is of the same form as the equation studied in the first part of this proof, the proof is complete.
\end{proof}

\end{document}